\documentclass[journal,twoside,web]{ieeecolor}
\usepackage{generic}
\usepackage{cite}
\usepackage{amsmath,amssymb,amsfonts}
\usepackage{graphicx}
\usepackage{textcomp}
\usepackage{mathtools}
\usepackage{bm}
\usepackage{tikz}
\usetikzlibrary{arrows.meta, positioning, calc}
\usepackage{hyperref}
\hypersetup{hidelinks}
\usepackage[dvipsnames]{xcolor}

\newtheorem{theorem}{Theorem}
\newtheorem{lemma}{Lemma}

\newtheorem{assumption}{Assumption}

\DeclareMathOperator{\Proj}{Proj}

\begin{document}

\title{Adaptive and Neural Operator Control \\ of Nonlinear Volterra Hyperbolic PDEs}
\author{Miroslav Krstic%
\thanks{M. Krstic is with the Department of Mechanical and Aerospace Engineering, University of California, San Diego, La Jolla, CA 92093 USA (e-mail: krstic@ucsd.edu).}%
\thanks{The principal AI aide in developing the paper was Claude.} }

\maketitle

\begin{abstract}
Adaptive control learns the plant online; neural-operator control learns the control gains offline. We bring the two together for a class of nonlinear hyperbolic PDEs whose dynamics are governed by an unknown Volterra series of arbitrarily many kernels. An observer-based passive identifier learns a truncation of this series online. The infinite-dimensional map that synthesizes the backstepping kernels from the parameter estimates --- a cascade of PDEs on simplex domains of increasing dimension, prohibitive to solve in real time --- is approximated once, offline, by a neural operator. The closed loop then carries two learning processes in series: online learning of the plant feeds an offline-learned PDE solver, whose output is the online control gains. We prove closed-loop stability and asymptotic regulation of the plant state, observer state, and input, on a basin that recovers the exact-kernel basin as the neural-operator accuracy improves. With a single Lyapunov function we absorb at once the perturbations---all vanishing---of truncating an infinite Volterra series, of identifying the plant online, and of approximating the gains. 
\end{abstract}

%==============================================================
\section{Introduction}\label{sec:intro}
%==============================================================

\IEEEPARstart{I}{n} analogy with the 1989 Sastry-Isidori classic on adaptive feedback linearization for ODEs \cite{SastryIsidori89}, this paper solves the same problem for boundary-controlled partial differential equations (PDEs), using PDE backstepping. PDE backstepping produces the feedback gains as the solution of a kernel PDE derived from a coordinate change to a stable target system. Two obstacles separate this methodology from deployment on an imperfectly known plant. First, the kernels are functionals of the plant coefficients and cannot be computed without a model. Second, even when the model is available, the kernel PDE must be re-solved whenever the plant estimate changes---as it does continuously under adaptation---and for many plant classes this repeated online computation is prohibitive. This paper removes both obstacles simultaneously for a class of nonlinear hyperbolic PDEs.

The plant is a first-order hyperbolic PDE whose in-domain nonlinearity is a Volterra series in the state, with kernels $f_n$ defined on the $n$-dimensional simplices $T_n$. When the kernels are known, this class is feedback-linearizable: a nonlinear Volterra backstepping transformation maps it to a linear transport equation that reaches the origin in finite time \cite{HV1}, and truncating the series at order $N$ yields a sup-norm stabilizer with an explicit basin of attraction \cite{HV2}, in the spirit of the parabolic Volterra designs of \cite{VK102}. The controller gains are the solution of a cascade of kernel PDEs, the $n$-th of which lives on the simplex $T_n\subset\mathbb{R}^n$; the domain dimension grows with the Volterra order, and the entire cascade must be re-solved for every kernel tuple. The second obstacle above thus appears here in an acute form.

Neural operators \cite{Lu2021Universal} address the second obstacle by learning, once and offline, the otherwise costly map from plant data to control, so that online a single forward evaluation replaces the PDE solve. The idea was introduced for PDE backstepping in \cite{KBS442} and brought to the truncated Volterra controller in \cite{HV3}, which---for a \emph{known} plant---learns the feedback operator itself and shows the approximation error to preserve closed-loop stability on a slightly reduced basin.

The first obstacle---the requirement of a model---is addressed by adaptive control. For hyperbolic PDEs, identifier-based designs estimate the plant coefficients online and pass the estimates to a certainty-equivalence controller \cite{AA19}. We adopt the observer-based \emph{passive} identifier introduced for ODEs in \cite{KK94,KrsticKKK1995}: the estimation error appears in the Lyapunov derivative as a dissipative term, so stability requires no persistence-of-excitation assumption.

Bringing the two ingredients together removes both obstacles, but couples them into a single architecture: the online identifier feeds the offline-learned solver, whose output is the online controller---two learning processes in series, one adapting the plant and one, pre-trained, mapping the plant estimate to the gains. This architecture was proposed in \cite{LBSK} for a linear recirculation benchmark. The present paper carries it to the fully nonlinear Volterra class, where a new difficulty arises: the linear kernel $f_1$ is itself unknown and estimated. In the known-model designs of \cite{HV1,HV2} the gain cascade is a sequence of closed-form quadratures; here the unknown $f_1$ makes each stage a linear Volterra integral equation that is self-coupled through $f_1$. We prove this equation invertible for every admissible $f_1$, with no smallness condition on the nonlinearity.

The contributions are as follows. (i) We design an observer-based passive identifier for the unknown Volterra kernels, together with an $N$-truncated certainty-equivalence controller acting on the observer state. (ii) For the exact-kernel closed loop we establish a priori $L^2$ stability and asymptotic regulation of the plant state, observer state, and input (Theorem~\ref{thm:stability}), with explicit bounds on the parameter and kernel estimates and on the input. (iii) We give a neural-operator implementation (Theorem~\ref{thm:neural}) that replaces the online kernel cascade by a single offline-trained surrogate and recovers the exact-kernel guarantees on a basin that expands to the exact-kernel basin as the approximation error vanishes. (iv) A numerical study illustrates the design.

The neural operator in this paper is not the one in the companion paper \cite{HV3}, and the two deserve to be told apart. In \cite{HV3} the plant is known, and what is learned is the truncated feedback operator itself, $\mathcal U_N$, carrying the plant's Volterra coefficients and the current state straight to the boundary input and folding the kernel-PDE solve and the nested Volterra integration into one surrogate; trained once over a class of admissible plants, it \emph{is} the controller, and its argument is a fixed plant. Here the plant is not known but identified online, and the estimate $\hat f(t)$ never stops moving, so no fixed state-to-input map will serve. What we learn is instead the kernel-generation map $\mathcal Q_N$, from the running estimate to the control gains; those gains are then paired with the observer state exactly, in the loop. In brief, \cite{HV3} learns the controller of a known plant; this paper learns the map that rebuilds the controller, online, as an unknown plant is discovered.

\emph{Notation.} A hat marks the quantities the algorithm forms and uses online: the parameter estimates $\hat f_n$, the observer state $\hat u$, the learned kernels $\hat k_n$ with the operator $\widehat{\mathcal Q}_N$ that produces them, and the controller $\hat K_N$ they define. A tilde marks an estimation error, $\tilde f_n:=f_n-\hat f_n$. A check marks the exact backstepping objects that are solved for in the analysis but never computed in the loop: the cascade kernels $\check k_n$, the transformation $\check K_N$ and its inverse $\check L_N$, and the operators $\check{\mathcal S}_n,\check\Psi_n$ built from them. The bilinear simplex operators $B^m_n$ and the symmetrization--permutation operators $D^{n,m}_j$ that assemble the kernel cascade are those of \cite{HV1}, used here as defined there.

The paper is organized as follows. Section~\ref{sec:plant} states the plant and standing assumption; Section~\ref{sec:design} presents the identifier and the truncated controller; Section~\ref{sec:theorem} states the two main theorems. Sections~\ref{sec:proof-identifier}--\ref{sec:proof-Barbalat} prove the exact-kernel result, and Sections~\ref{sec:QN-Lip}--\ref{sec:neural-closed-loop} prove the neural-operator result. Section~\ref{sec:sims} reports the simulation, and Section~\ref{sec:conclusions} concludes.
%==============================================================
\section{Plant and Standing Assumption}\label{sec:plant}
%==============================================================

The plant is the hyperbolic Volterra PDE
\begin{eqnarray}
u_t(x,t) &=& u_x(x,t)+F[u](x,t), \label{eq:plant1}\\
u(1,t) &=& U(t), \label{eq:plant2}\\
F[u](x,t) &=& \sum_{n\ge 1}\int_{T_n(x)}f_n(x,\xi_1,\ldots,\xi_n) \nonumber\\
&& \times\,\prod_{i=1}^n u(\xi_i,t)\,d\xi_n\cdots d\xi_1, \label{eq:plant3}
\end{eqnarray}
posed on $(x,t)\in[0,1)\times\mathbb{R}_+$, on the simplices $T_n(x)=\{(\xi_1,\ldots,\xi_n)\in\mathbb{R}^n:0\le\xi_n\le\cdots\le\xi_1\le x\}$. The Volterra series in \eqref{eq:plant3} starts at $n=1$. The linear kernel $f_1$ is one of the unknowns to be estimated.

\begin{assumption}[Factorial bound on plant kernels]\label{ass:factorial}
There exist known constants $D_f,\rho_f>0$ such that
\begin{equation}\label{eq:factorial}
\|f_n\|_{L^\infty(T_n(1))}\le \frac{n!\,D_f}{\rho_f^{\,n-1}},\qquad n\ge 1.
\end{equation}
\end{assumption}

This is Assumption~1 of \cite{HV1}, restated for $n\ge 1$. In particular $\|f_1\|_\infty\le D_f$.

%==============================================================
\section{Truncated Adaptive Controller}\label{sec:design}
%==============================================================

The controller estimates only the first $N$ kernels $f_1,\ldots,f_N$ (truncation order $N\ge 1$ is a design parameter), and uses an $N$-truncated Volterra series for the input $U(t)$.

The combination of the observer and the parameter update law below is a passive identifier: introduced for ODEs in \cite{KK94}, extended to linear hyperbolic PDEs by Anfinsen and Aamo \cite[\S4]{AA19}, and combined with a neural-operator approximation of the kernel-generation map for a linear hyperbolic benchmark by Lamarque, Bhan, Shi, and Krstic \cite[\S7]{LBSK}.

\subsection{Observer}\label{sec:observer}

Define the truncated estimated nonlinearity
\begin{equation}\label{eq:Fhat-N}
\hat F_N[u,\hat f](x,t):=\sum_{n=1}^N\int_{T_n(x)}\hat f_n(x,\xi,t)\prod_{i=1}^n u(\xi_i,t)\,d\xi,
\end{equation}
and the damping driver
\begin{equation}\label{eq:Phi-N}
\Phi_N(u)(t):=\sum_{n=1}^N\frac{\gamma_n}{n!}\,\|u(\cdot,t)\|_{L^2(0,1)}^{\,2n},
\end{equation}
a polynomial of degree $N$ in $\|u\|_{L^2}^2$, with the adaptation gains $\gamma_n>0$ as coefficients. The observer is
\begin{eqnarray}
\hat u_t(x,t) &=& \hat u_x(x,t) + \hat F_N[u,\hat f](x,t) \nonumber\\
&& +\, c_0\bigl(u(x,t)-\hat u(x,t)\bigr)\,\Phi_N(u)(t), \label{eq:obs1-N}\\
\hat u(1,t) &=& U(t), \label{eq:obs2-N}
\end{eqnarray}
where $c_0>0$ is the damping gain and $\hat u(\cdot,0)\in L^2(0,1)$ is arbitrary.
The observer state $\hat u$ is fed by the true state $u$ at $x\in[0,1)$ and shares the boundary control $U(t)$ at $x=1$. The estimation error is $e(x,t):=u(x,t)-\hat u(x,t)$.

\subsection{Parameter update laws}\label{sec:updates}

For $n=1,\ldots,N$, the parameter estimates $\hat f_n(x,\xi,t)\in L^\infty(T_n(1))$ evolve under
\begin{eqnarray}
\partial_t\hat f_n(x,\xi,t) &=& \Proj_{\bar F_n}\!\bigl(\tau_n(x,\xi,t),\,\hat f_n(x,\xi,t)\bigr), \label{eq:upd1-N}\\
\tau_n(x,\xi,t) &:=& \gamma_n\,e^{cx}\,e(x,t)\,\prod_{i=1}^n u(\xi_i,t), \label{eq:upd2-N}
\end{eqnarray}
where $\gamma_n>0$ is the order-$n$ adaptation gain and $c>0$ is a weighting exponent satisfying $c\le 4c_0$. The pointwise projection $\Proj_{\bar F_n}$ keeps $\hat f_n(\cdot,\cdot,t)$ inside the closed convex projection box
\begin{equation}\label{eq:box-N}
\bar F_n:=\Bigl\{\,g\in L^\infty(T_n(1))\;:\;\|g\|_\infty\le \bar f_n\,\Bigr\},\qquad \bar f_n:=\frac{n!\,D_f}{\rho_f^{\,n-1}},
\end{equation}
and satisfies, for $f^\star\in\bar F_n$, the standard inequality
\begin{eqnarray}
\lefteqn{-\bigl(f^\star(x,\xi)-g(x,\xi)\bigr)\Proj_{\bar F_n}(\tau,g)(x,\xi)} \nonumber\\
&\le& -\bigl(f^\star(x,\xi)-g(x,\xi)\bigr)\tau(x,\xi), \label{eq:Proj-ineq}
\end{eqnarray}
pointwise on $T_n(1)$. Initialization $\hat f_n(\cdot,\cdot,0)\in \bar F_n$ is part of the design.

We also define the Lipschitz-restricted projection box
\begin{eqnarray}
\bar F_n^{\rm Lip} &:=& \bigl\{g\in C^{0,1}(T_n(1))\,:\,\|g\|_\infty\le \bar f_n, \nonumber\\
&& \mathrm{Lip}(g)\le L_f\bigr\}\subset \bar F_n \label{eq:Bn-Lip}
\end{eqnarray}
for some prescribed $L_f>0$, used in the neural-operator design.

\subsection{Kernel integral equations}\label{sec:kernel-eqs}

At each time $t\ge 0$, given the current parameter estimate $\hat f(\cdot,\cdot,t)=(\hat f_1,\ldots,\hat f_N)$, the controller kernels $\check k_n(\cdot,\cdot,t)\in L^2(T_n(1))$, $n=1,\ldots,N$, are defined as the unique solution of the cascade of linear Volterra integral equations
\begin{eqnarray}
\check k_n(x,\xi_1,\ldots,\xi_n,t) &=& -\int_0^{\xi_n}\hat f_n(\xi^\sigma,t)\,d\sigma \nonumber\\
&& +\int_0^{\xi_n}\!B^1_n[\check k_n,\hat f_1](\xi^\sigma)\,d\sigma \nonumber\\
&& +\sum_{m=2}^n\int_0^{\xi_n}\!B^m_n[\check k_{n-m+1},\hat f_m](\xi^\sigma)\,d\sigma \label{eq:k-IE}
\end{eqnarray}
on $T_n(1)$, for $n=1,\ldots,N$.
Here the characteristic-coordinate shift is
\begin{equation}\label{eq:xi-sigma}
\xi^\sigma:=(x-\xi_n+\sigma,\,\xi_1-\xi_n+\sigma,\,\ldots,\,\xi_{n-1}-\xi_n+\sigma,\,\sigma),
\end{equation}
and $B^m_n[\kappa,g]$ for $m=1,\ldots,n$ is the bilinear simplex operator from \cite[eqs.~(20)--(21)]{HV1}: for $\kappa$ a function on $T_{n-m+1}(1)$ and $g$ a function on $T_m(1)$, $B^m_n[\kappa,g]$ is a function on $T_n(1)$ defined by
\begin{align}
B^m_n[\kappa,g]&(x,\xi_1,\ldots,\xi_n):=\sum_{j=1}^{n-m+1}\int_{\xi_{j+m-1}}^{\xi_{j-1}}\!\!D^{n,m}_j\bigl[ \nonumber\\
&\kappa(x,\xi_1,\ldots,\hat\xi_{j..j+m-1},s_j) \nonumber\\
&\times\,g(s_j,\xi_j,\ldots,\xi_{j+m-1})\bigr]\,ds_j \label{eq:Bmn-def}
\end{align}
with $\xi_0:=x$, $\hat\xi_{j..j+m-1}$ denoting omission of those indices, and $D^{n,m}_j$ the symmetrization-permutation operator of \cite[eq.~(21)]{HV1}.
Equation \eqref{eq:k-IE} is solved sequentially in $n$ as a cascade: at each $n$, the lower-order kernels $\check k_1,\ldots,\check k_{n-1}$ are already determined and the equation is a linear Volterra integral equation in $\check k_n$ alone. We denote by $\mathcal Q_N$ the kernel-generation operator
\begin{eqnarray}
\mathcal Q_N &:& \prod_{n=1}^N \bar F_n\;\longrightarrow\;\prod_{n=1}^N L^2(T_n(1)), \nonumber\\
\mathcal Q_N(\hat f) &:=& (\check k_1,\ldots,\check k_N), \label{eq:QN-def}
\end{eqnarray}
the map sending the parameter estimate tuple $\hat f=(\hat f_1,\ldots,\hat f_N)\in\prod_n \bar F_n$ to the unique solution $(\check k_1,\ldots,\check k_N)$ of the cascade \eqref{eq:k-IE}.

\subsection{Truncated control law}\label{sec:control-law}

The boundary control input is the $N$-truncated Volterra series
\begin{eqnarray}
U(t) &:=& \sum_{n=1}^N\int_{T_n(1)}\check k_n(1,\xi_1,\ldots,\xi_n,t) \nonumber\\
&& \times\,\prod_{i=1}^n\hat u(\xi_i,t)\,d\xi_n\cdots d\xi_1, \label{eq:U-def}
\end{eqnarray}
acting on the observer state $\hat u(\cdot,t)$, with kernels $\check k_n$ from \eqref{eq:k-IE}.

%==============================================================
\section{Main Results}\label{sec:theorem}
%==============================================================

The closed-loop norm of interest is
\begin{eqnarray}
\Pi(t) &:=& \Bigl(\|u(\cdot,t)\|^2_{L^2(0,1)}+\|\hat u(\cdot,t)\|^2_{L^2(0,1)} \nonumber\\
&& +\sum_{n=1}^N\|\tilde f_n(\cdot,\cdot,t)\|^2_{L^2(T_n(1))}\Bigr)^{1/2}, \label{eq:Pi-def}
\end{eqnarray}
where the parameter estimation error is
\begin{equation}\label{eq:ftilde-def}
\tilde f_n(x,\xi,t):=f_n(x,\xi)-\hat f_n(x,\xi,t),\qquad n=1,\ldots,N.
\end{equation}
Stability of the closed loop is asserted relative to the origin $\Pi=0$. 

\begin{theorem}[A priori $L^2$ stability and asymptotic regulation]\label{thm:stability}
Let Assumption~\ref{ass:factorial} hold, and let the design parameter $c$ satisfy $c\le 4c_0$. There exist constants $C_K^\infty,D_K^\infty,\mu_0>0$ depending on $D_f,\rho_f,N$, and a basin radius $\bar\Pi>0$, such that for any admissible solution $(u,\hat u,\hat f)$ of the closed-loop system \eqref{eq:plant1}--\eqref{eq:plant3}, \eqref{eq:obs1-N}--\eqref{eq:obs2-N}, \eqref{eq:upd1-N}--\eqref{eq:upd2-N}, \eqref{eq:k-IE}, \eqref{eq:U-def} on $[0,T)$ whose initial datum satisfies
\begin{equation}\label{eq:thm-IC}
\Pi(0)\le\bar\Pi,\qquad \hat f_n(\cdot,\cdot,0)\in \bar F_n,\;n=1,\ldots,N,
\end{equation}
the following hold for all $t\in[0,T)$:
\begin{enumerate}
\item \emph{Lyapunov stability of the origin in $\Pi$-norm:}
\begin{equation}\label{eq:thm-stability}
\Pi(t)\le \mu_0\,\Pi(0).
\end{equation}
\item \emph{Asymptotic regulation: plant state, observer state, and input all decay in $L^2$ (provided $T=\infty$).} As $t\to\infty$,
\begin{equation}\label{eq:thm-regulation}
\|u(\cdot,t)\|_{L^2}\to 0,\ \ \|\hat u(\cdot,t)\|_{L^2}\to 0,\ \ |U(t)|\to 0.
\end{equation}
%The parameter error states $\tilde f_n(t)$ remain \emph{bounded} for all $t\ge 0$ but \emph{need not} converge to zero (no persistence-of-excitation assumption is made).
\item \emph{Boundedness of parameter and kernel estimates:} for all $n=1,\ldots,N$, $t\ge 0$, and $x\in[0,1]$,
\begin{eqnarray}
\|\hat f_n(\cdot,\cdot,t)\|_{L^\infty(T_n(1))} &\le& \frac{n!\,D_f}{\rho_f^{\,n-1}}, \label{eq:thm-fhat-bound}\\
\|\check k_n(x,\cdot,t)\|_{L^1(T_n(x))} &\le& D_K^\infty\,(C_K^\infty)^{n-1}\,x^n. \label{eq:thm-kcheck-bound}
\end{eqnarray}
\item \emph{Boundedness of the input.} For all $t\ge 0$,
\begin{eqnarray}
|U(t)| &\le& \sum_{n=1}^N D_K^\infty\,(C_K^\infty\,\|\hat u(\cdot,t)\|_{L^2})^{n-1}\,\|\hat u(\cdot,t)\|_{L^2} \nonumber\\
&\le& \sum_{n=1}^N D_K^\infty\,(C_K^\infty\,\Pi(t))^{n-1}\,\Pi(t)<\infty. \label{eq:thm-U-bound}
\end{eqnarray}
\end{enumerate}
\end{theorem}

By an \emph{admissible solution} we mean a $(u,\hat u,\hat f)$ with $u,\hat u\in L^2(0,1)$ and $\hat f_n\in L^\infty(T_n(1))$ for which the pointwise update law \eqref{eq:upd1-N}--\eqref{eq:upd2-N} is satisfied a.e.\ on $T_n(1)$ at every $t$; in particular, the products $e(x,t)\prod_i u(\xi_i,t)$ admit a.e.-defined representatives on $T_n(1)$.

We next state the neural-operator-adaptive version of Theorem~\ref{thm:stability}. The kernels $\check k_n$ in the controller \eqref{eq:U-def} are obtained by solving the cascade integral equation \eqref{eq:k-IE} at every $t$; this is impractical for online implementation when $\hat f(t)$ varies in time. The neural-operator extension replaces this real-time cascade by a single offline-trained surrogate $\widehat{\mathcal Q}_N$ that maps $\hat f$ to (an approximation of) the kernel tuple $\check k$. The surrogate kernels are
\begin{equation}\label{eq:khat-thm}
\hat k_n(\cdot,\cdot,t):=\widehat{\mathcal Q}_N\bigl(\hat f(\cdot,\cdot,t)\bigr)_n,\qquad n=1,\ldots,N,
\end{equation}
and the implemented controller is
\begin{eqnarray}
U(t) &:=& \hat K_N[\hat u](1,t) \nonumber\\
&:=& \sum_{n=1}^N\int_{T_n(1)}\hat k_n(1,\xi,t)\prod_{i=1}^n\hat u(\xi_i,t)\,d\xi \label{eq:Uhat-thm}
\end{eqnarray}
in place of \eqref{eq:U-def}.

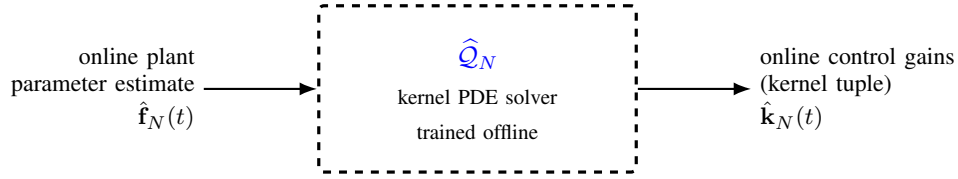
\begin{figure*}[t]
\centering
\begin{tikzpicture}[
  >=Latex, thick,
  qblock/.style={draw, rectangle, rounded corners=2pt, minimum height=22mm, minimum width=42mm, align=center, font=\normalsize, dashed, very thick},
  signal/.style={font=\small, inner sep=3pt}
]
\node[qblock] (Q) {{\color{blue}$\widehat{\mathcal Q}_N$}\\[3pt]{\footnotesize kernel PDE solver}\\[1pt]{\footnotesize trained offline}};
\node[signal, left=15mm of Q.west, align=right] (in) {online plant\\parameter estimate\\[2pt]$\hat{\mathbf f}_N(t)$};
\node[signal, right=15mm of Q.east, align=left] (out) {online control gains\\(kernel tuple)\\[2pt]$\hat{\mathbf k}_N(t)$};
\draw[->] (in.east) -- (Q.west);
\draw[->] (Q.east) -- (out.west);
\end{tikzpicture}
\caption{The neural operator $\widehat{\mathcal Q}_N$ is an offline-trained approximate solver of the kernel-generating cascade \eqref{eq:k-IE}, a sequence of PDEs on simplex domains $T_n(1)\subset\mathbb{R}^n$ of increasing dimension $n=1,\ldots,N$. The most computationally expensive component of the adaptive implementation --- the online solution of this PDE cascade at the running parameter estimate --- is pre-learned. Queried online at the identifier output $\hat{\mathbf f}_N(t)$, $\widehat{\mathcal Q}_N$ produces the online control gains $\hat{\mathbf k}_N(t)=\widehat{\mathcal Q}_N(\hat{\mathbf f}_N(t))$.}\label{fig:Qhat}
\end{figure*}

\begin{theorem}[Adaptive truncated neural backstepping]\label{thm:neural}
Let Assumption~\ref{ass:factorial} hold, and let the design parameter $c$ satisfy $c\le 4c_0$. Then there exist a neural-operator surrogate $\widehat{\mathcal Q}_N$ of $\mathcal Q_N$ on the Lipschitz-restricted projection box $\prod_n \bar F_n^{\rm Lip}$ defined in \eqref{eq:Bn-Lip} and a basin radius $\bar\Pi'>0$ such that, for any admissible solution of the closed-loop system \eqref{eq:plant1}--\eqref{eq:plant3}, \eqref{eq:obs1-N}--\eqref{eq:obs2-N}, \eqref{eq:upd1-N}--\eqref{eq:upd2-N}, \eqref{eq:Uhat-thm} with initial data $\hat f_n(\cdot,\cdot,0)\in \bar F_n^{\rm Lip}$ and $\Pi(0)\le\bar\Pi'$ that satisfies, additionally, the regularity hypothesis
\begin{equation}\label{eq:Lip-traj-hyp}
\hat f_n(\cdot,\cdot,t)\in \bar F_n^{\rm Lip},\qquad n=1,\ldots,N,\;t\ge 0,
\end{equation}
the conclusions \emph{(a)} (Lyapunov stability) and \emph{(b)} (asymptotic regulation) of Theorem~\ref{thm:stability} hold; in particular,
\begin{equation}\label{eq:thm-neural-reg}
\|u(\cdot,t)\|_{L^2}\to 0,\quad \|\hat u(\cdot,t)\|_{L^2}\to 0,\quad |U(t)|\to 0
\end{equation}
as $t\to\infty$. Moreover, the learned kernels $\hat k_n(\cdot,\cdot,t)=\widehat{\mathcal Q}_N(\hat f(t))_n$ and the input $U(t)$ remain uniformly bounded in $L^2(T_n(1))$ and $\mathbb R$, respectively, for all $t\ge 0$. 
%The exact cascade kernels $\check k_n$ defined by \eqref{eq:k-IE} appear in the analysis only as analytical comparison objects through $\Delta_n:=\check k_n-\hat k_n$; they are \emph{not} part of the implemented closed-loop dynamics.
\end{theorem}

%==============================================================
\section{Identifier Lyapunov Inequality}\label{sec:proof-identifier}
%==============================================================

We start by denoting the observer error as
\begin{equation}\label{eq:e-def}
e(x,t):=u(x,t)-\hat u(x,t).
\end{equation}
Subtracting \eqref{eq:obs1-N} from \eqref{eq:plant1} and using $u=\hat u+e$,
\begin{equation}\label{eq:e-dynamics}
e_t=e_x+\tilde F_N[u,\tilde f]+R_N[u]-c_0\,e\,\Phi_N(u),\qquad e(1,t)=0,
\end{equation}
where the truncated parameter-error feedforward is
\begin{eqnarray}
\tilde F_N[u,\tilde f](x,t) &:=& \sum_{n=1}^N\int_{T_n(x)}\tilde f_n(x,\xi,t) \nonumber\\
&& \times\,\prod_{i=1}^n u(\xi_i,t)\,d\xi, \label{eq:Ftilde-N}
\end{eqnarray}
and the unmodeled plant tail (the \emph{identifier truncation residual}) is
\begin{eqnarray}
R_N[u](x,t) &:=& \sum_{n>N}\int_{T_n(x)}f_n(x,\xi)\prod_{i=1}^n u(\xi_i,t)\,d\xi \nonumber\\
&=& F[u]-\hat F_N[u,f]. \label{eq:RNu-def}
\end{eqnarray}
Note: \eqref{eq:RNu-def} bounds independently of $\hat f$, since it depends only on the true $f_n$ for $n>N$.

We introduce the exponentially-weighted norm 
\begin{equation}\label{eq:c-norm}
\|e\|_c^2:=\int_0^1 e^{cx}e(x,t)^2\,dx
\end{equation}
and the identifier Lyapunov function 
\begin{equation}\label{eq:V-N}
V_{\rm iden}(t):=\tfrac12\,\|e(\cdot,t)\|_c^2+\sum_{n=1}^N\frac{1}{2\gamma_n}\,\|\tilde f_n(\cdot,\cdot,t)\|^2_{L^2(T_n(1))}.
\end{equation}

\begin{lemma}[Identifier $\dot V_{\rm iden}$ inequality with truncation residual]\label{lem:Vdot-N}
Let Assumption~\ref{ass:factorial} hold. Then along the closed-loop dynamics,
\begin{eqnarray}
\dot V_{\rm iden}(t) &\le& -\tfrac12\,e^2(0,t)-\tfrac{c}{2}\,\|e(\cdot,t)\|_c^2 \nonumber\\
&& -\,c_0\,\|e(\cdot,t)\|_c^2\,\Phi_N(u)(t) \nonumber\\
&& +\int_0^1 e^{cx}\,e(x,t)\,R_N[u](x,t)\,dx. \label{eq:Vdot-N}
\end{eqnarray}
Moreover the parameter update rate satisfies
\begin{equation}\label{eq:P4-N}
\sum_{n=1}^N\frac{1}{\gamma_n}\,\|\partial_t\hat f_n(\cdot,\cdot,t)\|^2_{L^2(T_n(1))}\le e^c\,\|e(\cdot,t)\|_c^2\,\Phi_N(u)(t),
\end{equation}
and the projection enforces $\hat f_n(\cdot,\cdot,t)\in \bar F_n$ for all $t\ge 0$, hence
\begin{equation}\label{eq:P1-N}
\|\hat f_n(\cdot,\cdot,t)\|_{L^\infty(T_n(1))}\le n!\,D_f/\rho_f^{n-1}.
\end{equation}
\end{lemma}

\begin{proof}
Compute
\begin{eqnarray}
\dot V_{\rm iden} &=& \underbrace{\int_0^1 e^{cx}e\,e_t\,dx}_{=I_1+I_2+I_R-c_0\|e\|_c^2\Phi_N(u)} \nonumber\\
&& -\sum_{n=1}^N\frac{1}{\gamma_n}\langle\tilde f_n,\partial_t\hat f_n\rangle, \label{eq:Vdot-expand}
\end{eqnarray}
where $e_t=e_x+\tilde F_N[u,\tilde f]+R_N[u]-c_0 e\Phi_N(u)$ from \eqref{eq:obs1-N}, and
\begin{eqnarray}
I_1 &:=& \int_0^1 e^{cx}e\,e_x\,dx, \nonumber\\
I_2 &:=& \int_0^1 e^{cx}e\,\tilde F_N[u,\tilde f]\,dx, \nonumber\\
I_R &:=& \int_0^1 e^{cx}e\,R_N[u]\,dx. \label{eq:I123-def}
\end{eqnarray}
 Since $e\,e_x=\tfrac12\partial_x(e^2)$, integration by parts in $x$ with $e(1,t)=0$ gives
\begin{eqnarray}
I_1 &=& \tfrac12\int_0^1 e^{cx}\partial_x(e^2)\,dx \nonumber\\
&=& \tfrac12\bigl[e^{cx}e^2\bigr]_0^1-\tfrac{c}{2}\|e\|_c^2=-\tfrac12 e^2(0,t)-\tfrac{c}{2}\|e\|_c^2. \label{eq:I1-N}
\end{eqnarray}
 The regressor $\tau_n=\gamma_n e^{cx}e\prod_i u(\xi_i)$ from \eqref{eq:upd2-N} is exactly the factor multiplying $\tilde f_n$ in $I_2$; integrating over the combined simplex $\{0\le\xi_n\le\cdots\le\xi_1\le x\le 1\}$,
\begin{equation}\label{eq:I2-N}
I_2=\sum_{n=1}^N\frac{1}{\gamma_n}\langle\tilde f_n,\tau_n\rangle_{L^2(T_n(1))}.
\end{equation}
 For each $n$, the standard projection property gives $-\tilde f_n\,\Proj_{\bar F_n}(\tau_n,\hat f_n)\le-\tilde f_n\,\tau_n$ pointwise on $T_n(1)$, hence
\begin{equation}\label{eq:proj-int-N}
-\sum_n\frac{1}{\gamma_n}\langle\tilde f_n,\partial_t\hat f_n\rangle\le-\sum_n\frac{1}{\gamma_n}\langle\tilde f_n,\tau_n\rangle=-I_2.
\end{equation}
 Substituting \eqref{eq:I1-N} and \eqref{eq:proj-int-N} into \eqref{eq:Vdot-expand}, the $I_2$-term cancels, leaving \eqref{eq:Vdot-N}.
 From $|\partial_t\hat f_n|\le|\tau_n|$ and the simplex symmetrization $\int_{T_n(x)}\prod u^2\,d\xi=\|u\|^{2n}_{L^2(0,x)}/n!$,
\begin{equation}\label{eq:tau-bound}
\|\partial_t\hat f_n\|^2_{L^2(T_n(1))}\le\|\tau_n\|^2_{L^2(T_n(1))}\le \gamma_n^2\,e^c\,\|e\|_c^2\,\frac{\|u\|^{2n}_{L^2}}{n!},
\end{equation}
and dividing by $\gamma_n$ and summing yields \eqref{eq:P4-N}. Inequality \eqref{eq:P1-N} is immediate from  projection.
\end{proof}

%==============================================================
\section{Kernel Cascade: Derivation and Bounds}\label{sec:proof-kernel}
%==============================================================

This section derives the kernel integral equations \eqref{eq:k-IE} from a backstepping matching condition and establishes the kernel $L^2$- and $L^1$-slice bounds in \eqref{eq:thm-kcheck-bound}.

The backstepping transformation map at frozen parameter estimate $\hat f$ is
\begin{eqnarray}
\check K_N[\hat u](x,t) &:=& \sum_{n=1}^N\int_{T_n(x)}\check k_n(x,\xi,t)\prod_{i=1}^n\hat u(\xi_i,t)\,d\xi, \nonumber\\
w(x,t) &:=& \hat u(x,t)-\check K_N[\hat u](x,t). \label{eq:K-formal}
\end{eqnarray}
By \eqref{eq:obs2-N} and \eqref{eq:U-def}, $w(1,t)=0$ identically. The boundary condition $w(1,t)=0$ is what the controller \eqref{eq:U-def} is \emph{designed to enforce}; the boundary condition $\check k_n|_{\xi_n=0}=0$ is a  property of the kernel PDE cascade.

\subsection{Matching identity}\label{sec:matching}

The following lemma establishes that the kernels defined by \eqref{eq:k-IE} satisfy a kernel PDE.

\begin{lemma}[Matching identity]\label{lem:matching}
The kernels $\check k_n(\cdot,\cdot,t)\in L^2(T_n(1))$, $n=1,\ldots,N$, defined by the integral equation \eqref{eq:k-IE}, satisfy the kernel PDE
\begin{eqnarray}
\partial_x\check k_n+\sum_{i=1}^n\partial_{\xi_i}\check k_n &=& -\hat f_n+B^1_n[\check k_n,\hat f_1] \nonumber\\
&& +\sum_{m=2}^n B^m_n[\check k_{n-m+1},\hat f_m], \nonumber\\
\check k_n|_{\xi_n=0} &=& 0,\quad n=1,\ldots,N. \label{eq:matching-PDE}
\end{eqnarray}
% Conversely, for smooth kernels $\check k_n\in C^1(T_n(1))$, the two formulations are equivalent, with \eqref{eq:matching-PDE} obtained from \eqref{eq:k-IE} by characteristic differentiation. Moreover, when the kernels $\check k_n$ satisfy \eqref{eq:matching-PDE} and the observer dynamics \eqref{eq:obs1-N} are run at $u\equiv\hat u$ and with damping turned off (so $\hat u_t=\hat u_x+\hat F_N[\hat u,\hat f]$), the backstepping image $w$ from \eqref{eq:K-formal} satisfies the pure transport equation $w_t=w_x$, $w(1,t)=0$.
\end{lemma}

\begin{proof}
%We treat the matching identity in two passes. The integrated form \eqref{eq:k-IE} is taken as the primary definition of $\check k_n$ at the $L^2$-level: it makes sense for $\check k_n\in L^2(T_n(1))$ without further regularity. The kernel PDE \eqref{eq:matching-PDE} is the formal differential counterpart: it holds rigorously when $\check k_n\in C^1(T_n(1))$ (or more generally has weak derivatives along the characteristic direction $(1,\ldots,1)$). 
The two formulations are equivalent for smooth kernels via integration along characteristics; for $L^2$-kernels we use \eqref{eq:k-IE} as the working definition and view \eqref{eq:matching-PDE} as the formal/distributional condition that motivates it. The proof goes in three steps: (1) the matching condition \eqref{eq:matching-PDE} is what \emph{makes} $w_t=w_x$ when $\hat u$ is frozen at $u\equiv\hat u$ and damping is off, by an explicit computation of $\partial_t \check K_N[\hat u]$ and $\partial_x \check K_N[\hat u]$; (2) the kernel PDE \eqref{eq:matching-PDE} is converted to integral form \eqref{eq:k-IE} by integration along the constant characteristic direction $(1,1,\ldots,1)$ of the differential operator $\partial_x+\sum_i\partial_{\xi_i}$; (3) the inflow boundary $\check k_n|_{\xi_n=0}=0$ is preserved by both forms.
Differentiating $\check K_N[\hat u]$ in $t$ at frozen $\hat f$ (so $\partial_t\check k_n=0$, since the kernels depend on $\hat f$ only) and substituting $\hat u_t=\hat u_x+\hat F_N[\hat u,\hat f]$ at the bilinear factor,
\begin{eqnarray}
\partial_t \check K_N[\hat u]\Big|_{\hat f\,\text{frozen}} &=& \sum_n\!\int_{T_n(x)}\!\!\check k_n\sum_j[\hat u_x(\xi_j) \nonumber\\
&& +\hat F_N[\hat u,\hat f](\xi_j)]\!\prod_{i\ne j}\!\hat u\,d\xi. \label{eq:partialt-K-formal}
\end{eqnarray}
Differentiating $\check K_N[\hat u]$ in $x$ at the moving boundary $\xi_1\le x$ and using $\check k_n|_{\xi_n=0}=0$ to annihilate the simplex-face boundary contributions,
\begin{eqnarray}
\partial_x \check K_N[\hat u] &=& \sum_n\!\int_{T_n(x)}\!\!\bigl[\partial_x\check k_n+\textstyle\sum_i\partial_{\xi_i}\check k_n\bigr]\prod\hat u\,d\xi \nonumber\\
&& +\sum_n\!\int_{T_n(x)}\!\!\check k_n\sum_j\hat u_x(\xi_j)\!\prod_{i\ne j}\!\hat u\,d\xi. \label{eq:partialx-K-formal}
\end{eqnarray}
Computing $w_t-w_x=\hat u_t-\hat u_x-\partial_t \check K_N[\hat u]+\partial_x \check K_N[\hat u]$ at $u\equiv\hat u$ and damping off:
\begin{eqnarray}
w_t-w_x &=& \hat F_N[\hat u,\hat f] \nonumber\\
&& -\sum_n\!\int_{T_n(x)}\!\!\bigl[\partial_x\check k_n+\textstyle\sum_i\partial_{\xi_i}\check k_n\bigr]\prod\hat u\,d\xi \nonumber\\
&& -\sum_n\!\int_{T_n(x)}\!\!\check k_n\sum_j\hat F_N[\hat u,\hat f](\xi_j)\!\prod_{i\ne j}\!\hat u\,d\xi. \label{eq:wt-wx-frozen}
\end{eqnarray}
The transport-derivative terms $\sum_n\int\check k_n\sum_j\hat u_x(\xi_j)\prod_{i\ne j}\hat u\,d\xi$ from the $t$- and $x$-pieces \emph{cancel pairwise}.
For $w_t=w_x$ to hold, the right-hand side of \eqref{eq:wt-wx-frozen} must vanish. Substituting $\hat F_N[\hat u,\hat f](\xi_j)=\sum_{m=1}^N\int_{T_m(\xi_j)}\hat f_m\prod\hat u\,d\eta$ into the bilinear $\check k_n\otimes\hat F_N$ piece at order $n$, and applying the symmetrization identity \cite[Lemma~A.1]{HV1} to express the resulting $(n+m-1)$-tensor in the standardized $T_{n+m-1}$-form, the order-$N'$ ($:=n+m-1$) matching gives
\begin{eqnarray}
\lefteqn{\bigl[\partial_x\check k_{N'}+\textstyle\sum_i\partial_{\xi_i}\check k_{N'}\bigr]+\hat f_{N'}-B^1_{N'}[\check k_{N'},\hat f_1]} \nonumber\\
&& -\sum_{m=2}^{N'}B^m_{N'}[\check k_{N'-m+1},\hat f_m]=0, \label{eq:matching-collected}
\end{eqnarray}
which is precisely \eqref{eq:matching-PDE}.
The differential operator $\partial_x+\sum_i\partial_{\xi_i}$ has constant characteristic direction $(1,1,\ldots,1)$. A point $(x,\xi_1,\ldots,\xi_n)\in T_n(1)$ is reached at characteristic-arclength $\xi_n$ from the inflow point $\xi^0:=(x-\xi_n,\xi_1-\xi_n,\ldots,\xi_{n-1}-\xi_n,0)$. Integrating \eqref{eq:matching-PDE} along the characteristic from $\sigma=0$ (where $\check k_n=0$ by the inflow BC) to $\sigma=\xi_n$,
\begin{eqnarray}
\check k_n(x,\xi) &=& \int_0^{\xi_n}\Bigl\{-\hat f_n(\xi^\sigma)+B^1_n[\check k_n,\hat f_1](\xi^\sigma) \nonumber\\
&& +\sum_{m=2}^n B^m_n[\check k_{n-m+1},\hat f_m](\xi^\sigma)\Bigr\}d\sigma, \label{eq:char-integration}
\end{eqnarray}
which is exactly \eqref{eq:k-IE}. Conversely, differentiating \eqref{eq:k-IE} along the characteristic gives back \eqref{eq:matching-PDE}.
The inflow boundary condition $\check k_n|_{\xi_n=0}=0$ is built into both \eqref{eq:k-IE} and \eqref{eq:matching-PDE}: in \eqref{eq:k-IE} the integral $\int_0^{\xi_n}$ vanishes at $\xi_n=0$; in \eqref{eq:matching-PDE} it is the BC stated. The  reason for this BC in the first place is that \eqref{eq:plant3} has no $u(0,t)$-trace (every Volterra integrand integrates $u$ on the open simplex away from $\xi=0$); matching the $u(0)$-trace pieces forces $\check k_n|_{\xi_n=0}=0$ at every order $n\ge 1$.
\end{proof}

\subsection{Kernel bounds}\label{sec:kernel-bounds}

For the analysis below we recast \eqref{eq:k-IE} in compact operator form. Equation \eqref{eq:k-IE} couples $\check k_n$ to itself through the $B^1_n[\check k_n,\hat f_1]$ term --- the novelty of having $f_1$ unknown --- and to the lower-order kernels $\check k_1,\ldots,\check k_{n-1}$ through the $B^m_n[\check k_{n-m+1},\hat f_m]$ terms for $m\ge 2$. Solving the cascade sequentially in $n$, at each step \eqref{eq:k-IE} is a linear Volterra integral equation in $\check k_n$ alone, of the form
\begin{equation}\label{eq:k-IE-compact}
(I-\check{\mathcal S}_n)\,\check k_n=\check\Psi_n,
\end{equation}
where $\check{\mathcal S}_n$ is the \emph{self-coupling operator at order $n$},
\begin{equation}\label{eq:Sn-def}
\check{\mathcal S}_n[\kappa](x,\xi_1,\ldots,\xi_n):=\int_0^{\xi_n}B^1_n[\kappa,\hat f_1](\xi^\sigma)\,d\sigma,
\end{equation}
acting on $L^2(T_n(1))$ and capturing how $\check k_n$ couples to itself through the unknown linear plant kernel $\hat f_1$ via the simplex operator $B^1_n$ integrated along the characteristic, and $\check\Psi_n$ is the \emph{cascade inhomogeneity at order $n$},
\begin{eqnarray}
\check\Psi_n(x,\xi_1,\ldots,\xi_n) &:=& -\int_0^{\xi_n}\!\hat f_n(\xi^\sigma,t)\,d\sigma \nonumber\\
&& +\sum_{m=2}^n\int_0^{\xi_n}B^m_n[\check k_{n-m+1},\hat f_m](\xi^\sigma)\,d\sigma, \label{eq:Psin-def}
\end{eqnarray}
collecting the $\hat f_n$-quadrature and the lower-order $\check k_{n-m+1}$ pieces (these are already determined in the cascade ordering, so $\check\Psi_n$ is known data at the $n$-th step).

For low truncation orders, \eqref{eq:k-IE} reads explicitly as follows. For $N=1$, the cascade has only one equation, a self-coupled scalar Volterra integral equation in $\check k_1(x,\xi,t)$ at frozen $\hat f_1$:
\begin{eqnarray}
\check k_1(x,\xi,t) &=& -\int_0^\xi\hat f_1(x-\xi+\sigma,\sigma,t)\,d\sigma \nonumber\\
&& +\int_0^\xi\!\int_{x-\xi+\sigma}^{x}\!\hat f_1(s,\sigma,t)\check k_1(s,x-\xi+\sigma,t)\,ds\,d\sigma \label{eq:kn1-IE}
\end{eqnarray}
For $N=2$, after solving \eqref{eq:kn1-IE} for $\check k_1$, the second equation reads
\begin{eqnarray}
\check k_2(x,\xi_1,\xi_2,t) &=& -\int_0^{\xi_2}\hat f_2(\xi^\sigma,t)\,d\sigma \nonumber\\
&& +\int_0^{\xi_2}B^1_2[\check k_2,\hat f_1](\xi^\sigma)\,d\sigma \nonumber\\
&& +\int_0^{\xi_2}B^2_2[\check k_1,\hat f_2](\xi^\sigma)\,d\sigma, \label{eq:kn2-IE}
\end{eqnarray}
with $\xi^\sigma=(x-\xi_2+\sigma,\xi_1-\xi_2+\sigma,\sigma)$ and $\check k_1$ already determined. The presence of the self-coupling $B^1_n[\check k_n,\hat f_1]$ means each $\check k_n$ is given by a Volterra integral equation, not by a closed-form quadrature as in \cite{HV1}; the structure is the hyperbolic analog of the parabolic cascade of \cite[eqs.~(60)--(63), Remark~6]{VK102}.

%The naive contraction estimate $\|\check{\mathcal S}_n\|_{L^2\to L^2}\le \binom{n+1}{2}D_f$ would require a smallness condition on $D_f$ to invert $I-\check{\mathcal S}_n$. This restriction is removed by a Volterra-iteration argument: the Neumann series for $(I-\check{\mathcal S}_n)^{-1}$ converges by a factorial $1/k!$ in the iteration count $k$, independent of $\|\check{\mathcal S}_n\|$. We state the bound in $L^2$, the natural setting for the cascade of \cite{HV1}.

\begin{lemma}[Volterra-iteration bound on $(I-\check{\mathcal S}_n)^{-1}$]\label{lem:Sn-inverse}
For each $n\ge 1$ and every $\hat f_1$ with $\|\hat f_1\|_\infty\le D_f$, the operator $\check{\mathcal S}_n$ defined by \eqref{eq:Sn-def} satisfies, with $M_n:=\binom{n+1}{2}D_f$,
\begin{equation}\label{eq:Sn-iter}
\|\check{\mathcal S}_n^k[\kappa]\|_{L^2(T_n(1))}\le \frac{M_n^k}{(k-1)!}\,\|\kappa\|_{L^2(T_n(1))},\ k\ge 1,
\end{equation}
hence the Neumann series converges and, for any $D_f>0$,
\begin{equation}\label{eq:Sn-inv-bound}
\|(I-\check{\mathcal S}_n)^{-1}\|_{L^2\to L^2}\le 1+M_n\,e^{M_n}.
\end{equation}
\end{lemma}

\begin{proof}
By \cite[Lemma~A.1]{HV1} the symmetrized simplex operator $B^1_n$ satisfies, for any $\kappa\in L^2(T_n(1))$ and $g\in L^\infty(T_1(1))$,
\begin{equation}\label{eq:B1n-L2}
\|B^1_n[\kappa,g]\|_{L^2(T_n(1))}\le \binom{n+1}{2}\,\|g\|_\infty\,\|\kappa\|_{L^2(T_n(1))},
\end{equation}
the bound following from the symmetrization producing $\binom{n+1}{2}$ pointwise terms each of which is a one-dimensional integral of $\kappa\cdot g$ along a simplex coordinate, bounded in $L^2$ by Cauchy--Schwarz. With $g=\hat f_1$ and $\|\hat f_1\|_\infty\le D_f$, the operator-norm bound is $\|B^1_n[\cdot,\hat f_1]\|_{L^2\to L^2}\le M_n:=\binom{n+1}{2}D_f$.
The following estimate captures the Volterra-iteration structure of $\check{\mathcal S}_n^k$ via characteristic coordinates, which expose $\check{\mathcal S}_n$ as a Volterra operator in $\xi_n$ acting on a slice norm in the transverse coordinates. Let $\eta_i:=\xi_i-\xi_n$ for $i=0,\ldots,n-1$ (with $\xi_0:=x$) and keep $\xi_n$ as the last coordinate; the simplex $T_n(1)$ in $(\eta,\xi_n)$ coordinates is a measurable subset of $[0,1-\xi_n]^n\times[0,1]$, and the change of variables is volume-preserving. In these coordinates the displacement $\xi^\sigma$ \emph{leaves $\eta$ unchanged and replaces $\xi_n$ by $\sigma$}: $\xi^\sigma_{(\eta,\xi_n)}=(\eta,\sigma)$. So $\check{\mathcal S}_n$ takes the simple Volterra form
\begin{equation}\label{eq:Sn-char-coords}
\check{\mathcal S}_n[\kappa](\eta,\xi_n)=\int_0^{\xi_n}B^1_n[\kappa,\hat f_1](\eta,\sigma)\,d\sigma,
\end{equation}
i.e., a Volterra integration in $\xi_n$ of the operator-valued kernel $B^1_n[\,\cdot\,,\hat f_1]$.
Iterate $k$ times: the $k$-fold $\sigma$-integration runs over the simplex $\Delta_k(\xi_n):=\{0\le\sigma_k\le\cdots\le\sigma_1\le\xi_n\}$ of measure $\xi_n^k/k!$, with the integrand $\bigl(B^1_n[\,\cdot\,,\hat f_1]\bigr)^k[\kappa](\eta,\sigma_k)$ standing for the $k$-fold composition of $B^1_n$ applied to $\kappa$ evaluated at $\sigma_k$ in the $\xi_n$-slot (each $B^1_n$ application carries its own integration in a simplex coordinate $s_j$, which we keep packaged inside the operator $B^1_n$):
\begin{equation}\label{eq:Sn-iter-form}
\check{\mathcal S}_n^k[\kappa](\eta,\xi_n)=\int_{\Delta_k(\xi_n)}(B^1_n)^k[\kappa](\eta,\sigma_k)\,d\sigma_1\cdots d\sigma_k.
\end{equation}
The operator-norm bound $\|(B^1_n)^k\|_{L^2\to L^2}\le M_n^k$ from \eqref{eq:B1n-L2} iterated $k$ times encapsulates all internal $s_j$-integrations.
By Cauchy--Schwarz on $\Delta_k(\xi_n)$,
\begin{eqnarray}
\bigl|\check{\mathcal S}_n^k[\kappa](\eta,\xi_n)\bigr|^2 &\le& |\Delta_k(\xi_n)|\int_{\Delta_k(\xi_n)}\bigl|(B^1_n)^k[\kappa](\eta,\sigma_k)\bigr|^2 d\sigma \nonumber\\
&=& \frac{\xi_n^k}{k!}\int_{\Delta_k(\xi_n)}\bigl|(B^1_n)^k[\kappa](\eta,\sigma_k)\bigr|^2 d\sigma. \label{eq:Sn-CS}
\end{eqnarray}
Integrate in $(\eta,\xi_n)\in T_n(1)$, swap the orders of integration (the variables $\sigma_1,\ldots,\sigma_{k-1}$ contribute a measure $(\xi_n-\sigma_k)^{k-1}/(k-1)!$ at fixed $\sigma_k\le\xi_n$), and use $\xi_n,(\xi_n-\sigma_k)\le 1$:
\begin{eqnarray}
\|\check{\mathcal S}_n^k[\kappa]\|_{L^2(T_n(1))}^2 &\le& \frac{1}{k!}\int_0^1\!\!\int_0^{\xi_n}\!\!\frac{(\xi_n-\sigma_k)^{k-1}}{(k-1)!}\,\xi_n^k \nonumber\\
&& \times\,\|(B^1_n)^k[\kappa](\cdot,\sigma_k)\|_{L^2_\eta}^2\,d\sigma_k\,d\xi_n \nonumber\\
&\le& \frac{1}{k!\,(k-1)!}\int_0^1\!\!\|(B^1_n)^k[\kappa](\cdot,\sigma_k)\|_{L^2_\eta}^2\,d\sigma_k \nonumber\\
&& =\frac{1}{k!\,(k-1)!}\,\|(B^1_n)^k[\kappa]\|_{L^2(T_n(1))}^2. \label{eq:Sn-k-final}
\end{eqnarray}
Combined with $\|(B^1_n)^k[\kappa]\|_{L^2}\le M_n^k\|\kappa\|_{L^2}$ from \eqref{eq:B1n-L2} iterated $k$ times,
\begin{eqnarray}
\|\check{\mathcal S}_n^k[\kappa]\|_{L^2(T_n(1))} &\le& \frac{M_n^k}{\sqrt{k!\,(k-1)!}}\,\|\kappa\|_{L^2(T_n(1))} \nonumber\\
&\le& \frac{M_n^k}{(k-1)!}\,\|\kappa\|_{L^2(T_n(1))}, \label{eq:Sn-iter-bound}
\end{eqnarray}
using $k!\ge(k-1)!$. The Neumann series
\begin{equation}\label{eq:Neumann-series}
\|(I-\check{\mathcal S}_n)^{-1}\|_{L^2\to L^2}\le 1+\sum_{k\ge 1}\frac{M_n^k}{(k-1)!}=1+M_n\,e^{M_n}
\end{equation}
converges for any $D_f>0$.
\end{proof}

\begin{lemma}[Kernel cascade slice bounds]\label{lem:k-bounds}
Under Assumption~\ref{ass:factorial}, the kernels $\check k_n$ defined by \eqref{eq:k-IE} for any $\hat f$ satisfying \eqref{eq:thm-fhat-bound} satisfy, for $n=1,\ldots,N$, $t\ge 0$, and $x\in[0,1]$,
\begin{eqnarray}
\|\check k_n(x,\cdot,t)\|_{L^2(T_n(x))} &\le& \sqrt{n!}\,D_K^\infty\,(C_K^\infty)^{n-1}\,x^{n/2}, \label{eq:k-L2-explicit}\\
\|\check k_n(x,\cdot,t)\|_{L^1(T_n(x))} &\le& D_K^\infty\,(C_K^\infty)^{n-1}\,x^n, \label{eq:k-L1-explicit}
\end{eqnarray}
where $D_K^\infty,C_K^\infty$ are positive constants depending on $D_f,\rho_f,N$, of the form
\begin{equation}\label{eq:DK-CK-infty}
D_K^\infty\le e^{\binom{N+1}{2}D_f}\,D_K,\qquad C_K^\infty\le e^{\binom{N+1}{2}D_f}\,C_K,
\end{equation}
where $D_K,C_K>0$ are the deterministic $L^2$-slice cascade constants from \cite[Theorem~1]{HV1} (the $L^2$-slice bound for the $f_1\equiv 0$ case).
\end{lemma}

\begin{proof}
\emph{$L^2$-slice bound.} Apply Lemma~\ref{lem:Sn-inverse} to the integral equation \eqref{eq:k-IE}: $\|\check k_n\|_{L^2(T_n(1))}\le e^{\binom{n+1}{2}D_f}\|\check\Psi_n\|_{L^2(T_n(1))}$, where $\check\Psi_n$ collects the inhomogeneous terms (the $-\int_0^{\xi_n}\hat f_n$ piece and the $\sum_{m\ge 2}B^m_n[\check k_{n-m+1},\hat f_m]$ pieces). The inhomogeneity at order $n$ involves $\hat f_n$ and lower-order kernels $\check k_{n-m+1}$ for $m\ge 2$, with $L^2$-norms bounded by the recursion \cite[eq.~(8) and proof of Theorem~1]{HV1}:
\begin{eqnarray}
\|\check\Psi_n\|_{L^2(T_n(1))} &\le& \sqrt{n!}\,D_f/\rho_f^{n-1} \nonumber\\
&& +\sum_{m=2}^n c^{n,m}\,\|\check k_{n-m+1}\|_{L^2(T_{n-m+1}(1))} \nonumber\\
&& \quad\times\|\hat f_m\|_\infty, \label{eq:Psi-L2-bound}
\end{eqnarray}
with combinatorial coefficients $c^{n,m}$ from the symmetrization. By induction in $n$, with $C_K^\infty$ absorbing the $e^{\binom{n+1}{2}D_f}$ factor at each step (and using the $L^2$-cascade structure of \cite{HV1}), we obtain \eqref{eq:k-L2-explicit}.

\emph{$L^1$-slice bound.} By Cauchy--Schwarz on $T_n(x)$,
\begin{eqnarray}
\|\check k_n(x,\cdot,t)\|_{L^1(T_n(x))} &\le& \|\check k_n(x,\cdot,t)\|_{L^2(T_n(x))}\cdot\sqrt{\mathrm{vol}(T_n(x))} \nonumber\\
&=& \|\check k_n(x,\cdot,t)\|_{L^2(T_n(x))}\cdot\frac{x^{n/2}}{\sqrt{n!}}. \label{eq:HV2-bridge}
\end{eqnarray}
Substituting \eqref{eq:k-L2-explicit}, the $\sqrt{n!}$ cancels:
\begin{eqnarray}
\|\check k_n(x,\cdot,t)\|_{L^1(T_n(x))} &\le& \sqrt{n!}\,D_K^\infty(C_K^\infty)^{n-1}x^{n/2}\cdot\frac{x^{n/2}}{\sqrt{n!}} \nonumber\\
&=& D_K^\infty(C_K^\infty)^{n-1}x^n, \label{eq:k-L1-derivation}
\end{eqnarray}
which is \eqref{eq:k-L1-explicit}.
\end{proof}

%==============================================================
\begin{figure}[t]
\centering
\begin{tikzpicture}[
  font=\footnotesize,
  box/.style={draw, thick, rounded corners=2pt, align=center, inner sep=3.5pt},
  thm/.style={draw, semithick, rounded corners=2pt, fill=black!7, align=center,
              inner sep=3pt, font=\footnotesize\bfseries},
  sec/.style={font=\scriptsize, text=black!60, inner sep=1.5pt},
  a/.style={-{Latex[length=1.6mm]}, semithick},
]
\node[box] (tgt) {TARGET\quad $w_t=w_x+$ five perturbations};
\node[sec, above=0.3mm of tgt.north east, anchor=south east] {Sec~VII};
\node[box, below=5mm of tgt] (v) {$\dot V_{\rm tot}\le 0$};
\node[sec, right=1mm of v] {Sec~IX};
\node[thm, below=8mm of v] (t1) {THEOREM 1};
\node[thm, below=14mm of t1] (t3) {THEOREM 3};
\node[thm, right=22mm of t3] (t2) {THEOREM 2};

\draw[a] (tgt) -- node[sec, right] {dominate the five \ (Sec~VIII)} (v);
\draw[a] (v) -- (t1);
\draw[a] (t1) -- node[sec, right, align=left]
      {approximate controller\\ in place of the exact one,\\ same analysis (Sec~XI--XIII)} (t3);
\draw[a] (t3) -- node[sec, midway, above] {approx.\ error $\le\varepsilon$} (t2);

\node[sec, left=9mm of v, align=center] (id) {identifier\\ energy\\ (Sec~V)};
\draw[a] (id) -- (v);
\node[sec, right=2mm of t1] {Barbalat (Sec~X) $\Rightarrow$ regulation};
\node[sec, below=1mm of t2] {basin $\to$ exact ($\varepsilon\to 0$)};
\end{tikzpicture}
\caption{Architecture of the stability analysis. }
\label{fig:roadmap}
\end{figure}
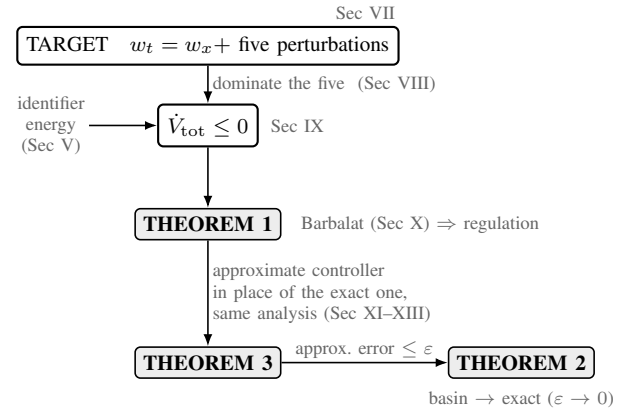

\section{Target System and Backstepping Transformation}\label{sec:proof-target}
%==============================================================

Figure~\ref{fig:roadmap} charts the analysis of Sections~\ref{sec:proof-target}--\ref{sec:neural-closed-loop}.

The backstepping image of $\hat u$ is
\begin{eqnarray}
w(x,t) &:=& \hat u(x,t)-\check K_N[\hat u](x,t), \nonumber\\
\check K_N[\hat u] &:=& \sum_{n=1}^N\int_{T_n(x)}\check k_n(x,\xi,t)\prod_{i=1}^n\hat u(\xi_i,t)\,d\xi. \label{eq:w-def}
\end{eqnarray}

\begin{lemma}[Perturbed target PDE]\label{lem:target}
Along the closed-loop dynamics \eqref{eq:plant1}--\eqref{eq:plant3}, \eqref{eq:obs1-N}--\eqref{eq:obs2-N}, \eqref{eq:U-def}, the backstepping image \eqref{eq:w-def} satisfies
\begin{eqnarray}
w_t &=& w_x+(I-\check K_N^*)R^u_N \nonumber\\
&& +\,c_0\Phi_N(u)\,(I-\check K_N^*)e-\Omega+R^K_N, \label{eq:tgt1-clean}\\
w(1,t) &=& 0, \label{eq:tgt2-clean}
\end{eqnarray}
on $(x,t)\in[0,1)\times\mathbb{R}_+$, where:
\begin{itemize}
\item the \emph{plant--observer feedforward mismatch} is
\begin{eqnarray}
R^u_N(x,t) &:=& \hat F_N[u,\hat f](x,t)-\hat F_N[\hat u,\hat f](x,t) \nonumber\\
&=& \sum_{n=1}^N\int_{T_n(x)}\hat f_n(x,\xi,t) \nonumber\\
&& \times\bigl[\textstyle\prod_i u(\xi_i)-\prod_i\hat u(\xi_i)\bigr]\,d\xi, \label{eq:Ru-N-def}
\end{eqnarray}
linear in $e=u-\hat u$ by the multilinear telescoping identity \cite[eq.~(37)]{HV2};
\item the \emph{kernel-time-derivative drift} is
\begin{equation}\label{eq:Omega-def}
\Omega(x,t):=\sum_{n=1}^N\int_{T_n(x)}\partial_t\check k_n(x,\xi,t)\prod_{i=1}^n\hat u(\xi_i,t)\,d\xi;
\end{equation}
\item the \emph{controller-truncation residual}, collecting the unmatched-order pieces of $\partial_x \check K_N[\hat u]$ at total order $n+m-1>N$, is
\begin{eqnarray}
R^K_N(x,t) &:=& \sum_{\substack{n+m-1>N\\ 1\le n,m\le N}}\int_{T_n(x)}\check k_n(x,\xi,t)\sum_{j=1}^n \nonumber\\
&& \times\int_{T_m(\xi_j)}\!\hat f_m(\xi_j,\eta,t)\prod_{i=1}^m\!\hat u(\eta_i,t)\,d\eta \nonumber\\
&& \times\prod_{i\ne j}\!\hat u(\xi_i,t)\,d\xi; \label{eq:RKN-def}
\end{eqnarray}
\item the \emph{simplex-tensor operator} $(I-\check K_N^*)$ acting on any function $\phi$ on $[0,1]$ is
\begin{eqnarray}
(I-\check K_N^*)\phi(x,t) &:=& \phi(x,t)-\sum_{n=1}^N\int_{T_n(x)}\check k_n(x,\xi,t) \nonumber\\
&& \times\sum_{j=1}^n\phi(\xi_j,t)\!\prod_{i\ne j}\!\hat u(\xi_i,t)\,d\xi. \label{eq:KN-star-def}
\end{eqnarray}
\end{itemize}
The parameter error $\tilde f$ does not appear directly in \eqref{eq:tgt1-clean}; it enters only through $\Omega$, since by the update law \eqref{eq:upd1-N} $\partial_t\hat f$ is itself state-multiplied.
\end{lemma}

\begin{proof}
Differentiate \eqref{eq:w-def} in $t$ and substitute the observer dynamics \eqref{eq:obs1-N}:
\begin{equation}\label{eq:wt-step}
w_t=\hat u_x+\hat F_N[u,\hat f]+c_0 e\,\Phi_N(u)-\partial_t \check K_N[\hat u].
\end{equation}
With the actual closed-loop $\hat u_t=\hat u_x+\hat F_N[u,\hat f]+c_0 e\,\Phi_N(u)$ substituted at each $\xi_j$,
\begin{eqnarray}
\partial_t \check K_N[\hat u] &=& \Omega+\sum_n\!\int_{T_n(x)}\!\!\check k_n\sum_j[\hat u_x(\xi_j)+\hat F_N[u,\hat f](\xi_j) \nonumber\\
&& +\,c_0 e(\xi_j)\Phi_N(u)]\prod_{i\ne j}\hat u\,d\xi. \label{eq:ptKN-step}
\end{eqnarray}
The matching identity (Lemma~\ref{lem:matching}) for $\check k_n$ at frozen $\hat f$ gives
\begin{eqnarray}
\partial_x \check K_N[\hat u] &=& -\hat F_N[\hat u,\hat f] \nonumber\\
&& +\sum_n\!\int_{T_n(x)}\!\!\check k_n\sum_j\hat F_N[\hat u,\hat f](\xi_j)\!\prod_{i\ne j}\!\hat u\,d\xi \nonumber\\
&& +\sum_n\!\int_{T_n(x)}\!\!\check k_n\sum_j\hat u_x(\xi_j)\!\prod_{i\ne j}\!\hat u\,d\xi. \label{eq:pxKN-step}
\end{eqnarray}
Form $w_t-w_x$. Since $w_x=\hat u_x-\partial_x\check K_N[\hat u]$, subtracting from \eqref{eq:wt-step} gives
\begin{equation}\label{eq:wt-wx-assemble}
w_t-w_x=\hat F_N[u,\hat f]+c_0 e\,\Phi_N(u)-\partial_t\check K_N[\hat u]+\partial_x\check K_N[\hat u].
\end{equation}
Insert \eqref{eq:ptKN-step} and \eqref{eq:pxKN-step} and collect by structure. The transport-derivative terms $\sum_n\int\check k_n\sum_j\hat u_x(\xi_j)\prod_{i\ne j}\hat u\,d\xi$ enter from $\partial_t\check K_N$ and $\partial_x\check K_N$ with opposite signs and cancel. The standalone feedforward is $\hat F_N[u,\hat f]-\hat F_N[\hat u,\hat f]=R^u_N$, while the bilinear terms contribute $-\sum_n\int\check k_n\sum_j R^u_N(\xi_j)\prod_{i\ne j}\hat u\,d\xi$; by the definition \eqref{eq:KN-star-def} of $(I-\check K_N^*)$ these combine into $(I-\check K_N^*)R^u_N$, with the unmatched orders $n+m-1>N$ left over as $R^K_N$. The damping terms combine the same way into $c_0\Phi_N(u)(I-\check K_N^*)e$, and the $\partial_t\check k_n$ terms give $-\Omega$. This is \eqref{eq:tgt1-clean}; the boundary condition $w(1,t)=0$ follows from $\hat u(1,t)=U(t)=\check K_N[\hat u](1,t)$.
\end{proof}

%==============================================================
\section{Dominance Bounds on Target-System Perturbations}\label{sec:proof-dominance}
%==============================================================

The controller Lyapunov component is
\begin{equation}\label{eq:W-def}
V_{\rm targ}(t):=\tfrac12\|w(\cdot,t)\|_c^2.
\end{equation}
Abbreviating the observer nonlinear damping term from \eqref{eq:tgt1-clean} as
\begin{equation}\label{eq:D-def}
D(x,t):=c_0\Phi_N(u)\,(I-\check K_N^*)e,
\end{equation}
computing $\dot V_{\rm targ}$ along \eqref{eq:tgt1-clean}--\eqref{eq:tgt2-clean}, and integrating by parts in $x$ with $w(1,t)=0$,
\begin{eqnarray}
\dot V_{\rm targ} &=& -\tfrac12 w^2(0,t)-\tfrac{c}{2}\|w\|_c^2 \nonumber\\
&& +\langle e^{cx}w,(I-\check K_N^*)R^u_N\rangle+\langle e^{cx}w,D\rangle \nonumber\\
&& -\langle e^{cx}w,\Omega\rangle+\langle e^{cx}w,R^K_N\rangle. \label{eq:Wdot-expand}
\end{eqnarray}
We use the inverse transformation $\hat u=w+\check L_N[w]$, with $\check L_N$ the inverse of $\check K_N$ at frozen $\hat f$ (\cite[Lemma~10]{HV2}), giving
\begin{equation}\label{eq:inverse-transform}
\|\hat u\|_{L^2}\le G_2\,\|w\|_c,\qquad \|w\|_{L^2}\le H_2\,\|\hat u\|_{L^2},
\end{equation}
with $G_2,H_2>0$ depending on $D_K^\infty,C_K^\infty,N$. The multilinear telescoping identity \cite[eq.~(37)]{HV2} splits the mismatch,
\begin{equation}\label{eq:Ru-split}
R^u_N=R^{u,1}_N+\sum_{n\ge 2}R^{u,n}_N,
\end{equation}
into the genuinely bilinear $n=1$ piece and the higher-order pieces in the state.

\begin{lemma}[Dominance bounds]\label{lem:dominance}
On the basin $\|\hat u\|_{L^2}\le\bar u<\min(\rho_f,1/C_K^\infty)$, the cross-terms in \eqref{eq:Wdot-expand} together with the identifier residual $\langle e^{cx}e,R_N[u]\rangle$ from \eqref{eq:Vdot-N} satisfy, for any $\delta>0$ and $\lambda>0$,
\begin{eqnarray}
|\langle e^{cx}w,\Omega\rangle| &\le& \frac{2(L^\Omega)^2}{\lambda c}\|w\|_c^2+\frac{\lambda c}{8}\|e\|_c^2\Phi_N(u), \label{eq:Omega-Y}\\
|\langle e^{cx}w,D\rangle| &\le& \frac{c_0\Phi_N(u)(1+e^{c/2}\Lambda^D_K)^2}{2\lambda}\|w\|_c^2 \nonumber\\
&& +\frac{\lambda c_0}{2}\|e\|_c^2\Phi_N(u), \label{eq:D-Y-corrected}\\
|\langle e^{cx}w,R^K_N\rangle| &\le& \frac{c}{8}\|w\|_c^2 \nonumber\\
&& +\frac{2e^c(L^{R^K}_N)^2}{c(1-C_K^\infty\bar u)^2}\,\|\hat u\|_{L^2}^{2(N+1)}, \label{eq:RK-Y}\\
\Bigl|\Bigl\langle e^{cx}w,\sum_{n\ge 2}R^{u,n}_N\Bigr\rangle\Bigr| &\le& \frac{2e^c\Lambda^{u,\ge 2}_N(\bar u)^2}{\lambda c}\|w\|_c^2+\frac{\lambda c}{8}\|e\|_c^2, \label{eq:Ru-ge2-Y}\\
|\langle e^{cx}w,R^{u,1}_N\rangle| &\le& \frac{D_f^2 e^c}{2\delta}\|w\|_c^2+\frac{\delta}{2}\|e\|_c^2, \label{eq:Ru-1-Y}\\
\Bigl|\int_0^1 e^{cx}e\,R_N[u]\,dx\Bigr| &\le& \frac{c}{8}\|e\|_c^2 \nonumber\\
&& +\frac{2e^c(L^R_N)^2}{c(1-\bar u/\rho_f)^2}\,\|u\|_{L^2}^{2(N+1)}, \label{eq:RN-Y}
\end{eqnarray}
where
\begin{eqnarray}
\Lambda^\Omega_N(r) &:=& D_K^\Omega\sum_{n=1}^N\sqrt{n}\,(C_K^\Omega r)^{n-1}\,r,\quad r<1/C_K^\Omega, \label{eq:LambdaOmega}\\
\Lambda^D_K(r) &:=& D_K^\infty\sum_{n=1}^N n^{3/2}\,(C_K^\infty r)^{n-1},\quad r<1/C_K^\infty, \label{eq:LambdaDK-def}\\
\Lambda^{u,\ge 2}_N(r) &:=& D_f\sum_{n=2}^N n^2\,(r/\rho_f)^{n-1},\quad r<\rho_f, \label{eq:Lambda-uge2}
\end{eqnarray}
and
\begin{equation}\label{eq:LOmega-LR-def}
L^\Omega:=e^{c/2}\sqrt{\gamma^{\max}e^c}\,\Lambda^\Omega_N(\|\hat u\|_{L^2}),\qquad L^R_N:=D_f/\rho_f^N,
\end{equation}
while $L^{R^K}_N$ is a constant from the cascade structure.
\end{lemma}

\begin{proof}
Each cross-term below is reduced by the weighted Young inequality: for $a,b\ge 0$ and any $\vartheta>0$,
\begin{equation}\label{eq:young}
ab\le\frac{1}{2\vartheta}\,a^2+\frac{\vartheta}{2}\,b^2.
\end{equation}
In every case $a$ is a constant multiple of $\|w\|_c$, charged to the $\|w\|_c^2$ dissipation of $\dot V_{\rm targ}$, and $b$ is a multiple of $\|e\|_c$ or $\|e\|_c\sqrt{\Phi_N(u)}$, charged to the $\|e\|_c^2$ or $\|e\|_c^2\Phi_N(u)$ dissipation of $\lambda\dot V_{\rm iden}$; the free parameter $\vartheta$ apportions the split.

\emph{(i) The drift term $\Omega$.} The drift $\Omega(x,t):=\sum_{n=1}^N\int_{T_n(x)}\partial_t\check k_n(x,\xi,t)\prod_i\hat u(\xi_i,t)\,d\xi$ from \eqref{eq:Omega-def} requires a bound on $\partial_t\check k_n$. Differentiating \eqref{eq:k-IE} in $t$,
\begin{equation}\label{eq:lin-cascade}
(I-\check{\mathcal S}_n)\,\partial_t\check k_n=\partial_t\check\Psi_n+(\partial_t\check{\mathcal S}_n)\check k_n,
\end{equation}
where
\begin{equation}\label{eq:dt-S-def}
\partial_t\check{\mathcal S}_n[\kappa]:=\int_0^{\xi_n}B^1_n[\kappa,\partial_t\hat f_1](\xi^\sigma)\,d\sigma,
\end{equation}
and
\begin{eqnarray}
\partial_t\check\Psi_n &=& -\int_0^{\xi_n}\!\partial_t\hat f_n(\xi^\sigma)\,d\sigma \nonumber\\
&& +\sum_{m=2}^n\int_0^{\xi_n}\bigl[B^m_n[\partial_t\check k_{n-m+1},\hat f_m] \nonumber\\
&& +B^m_n[\check k_{n-m+1},\partial_t\hat f_m]\bigr](\xi^\sigma)\,d\sigma. \label{eq:Psidot}
\end{eqnarray}
By Lemma~\ref{lem:Sn-inverse}, $\|(I-\check{\mathcal S}_n)^{-1}\|_{L^2\to L^2}\le e^{\binom{n+1}{2}D_f}$. Combining with $L^2$-bounds on the $B^m_n$-pieces and the $L^2$-cascade structure of \cite{HV1}, induction on \eqref{eq:lin-cascade} yields constants $D_K^\Omega,C_K^\Omega$ depending on $D_f$ such that
\begin{eqnarray}
\|\partial_t\check k_n(\cdot,\cdot,t)\|_{L^2(T_n(1))} &\le& D_K^\Omega\,(C_K^\Omega)^{n-1}\sqrt{n!} \nonumber\\
&& \times\sum_{m=1}^n\|\partial_t\hat f_m(\cdot,\cdot,t)\|_{L^2(T_m(1))}. \label{eq:partialt-k-L2}
\end{eqnarray}
By Cauchy--Schwarz on $T_n(x)$ with the simplex bridge $\|\prod\hat u\|_{L^2(T_n(x))}\le\|\hat u\|^n_{L^2}/\sqrt{n!}$ \cite{HV2}, the $\sqrt{n!}$ factors cancel, leaving
\begin{equation}\label{eq:Omega-x-bd}
|\Omega(x,t)|\le\Lambda^\Omega_N(\|\hat u\|_{L^2})\cdot\sqrt{\sum_{m=1}^N\|\partial_t\hat f_m\|^2_{L^2(T_m(1))}}.
\end{equation}
Pairing against $e^{cx}w$ by Cauchy--Schwarz and using \eqref{eq:Omega-x-bd}, then \eqref{eq:P4-N} to convert the update rate $\sqrt{\sum_m\|\partial_t\hat f_m\|^2}\le\sqrt{\gamma^{\max}e^c}\,\|e\|_c\sqrt{\Phi_N(u)}$,
\begin{equation}\label{eq:Omega-pre-Young}
|\langle e^{cx}w,\Omega\rangle|\le\|w\|_c\,\|\Omega\|_c\le L^\Omega\,\|w\|_c\,\|e\|_c\sqrt{\Phi_N(u)},
\end{equation}
with $L^\Omega$ as in \eqref{eq:LOmega-LR-def}. Young's inequality \eqref{eq:young} with $\vartheta=\lambda c/4$ on $a=L^\Omega\|w\|_c$, $b=\|e\|_c\sqrt{\Phi_N(u)}$ then gives \eqref{eq:Omega-Y}.

\emph{(ii) The damping term $D$.} Split $D=c_0\Phi_N(u)\,e+D_2$ where
\begin{eqnarray}
D_2(x,t) &:=& -c_0\Phi_N(u)\sum_{n=1}^N\int_{T_n(x)}\check k_n(x,\xi,t) \nonumber\\
&& \times\sum_{j=1}^n e(\xi_j,t)\!\prod_{i\ne j}\!\hat u(\xi_i,t)\,d\xi. \label{eq:D2-def}
\end{eqnarray}
By Cauchy--Schwarz, $|\langle e^{cx}w,c_0\Phi_N(u)\,e\rangle|\le c_0\Phi_N(u)\|w\|_c\|e\|_c$. At each order $n$, by the $\sum_j$ symmetry pick the $j=1$ slot as representative; by Cauchy--Schwarz on $T_n(x)$ with $\|\check k_n(x,\cdot)\|_{L^2(T_n(x))}\le\sqrt{n!}\,D_K^\infty(C_K^\infty)^{n-1}x^{n/2}$ (Lemma~\ref{lem:k-bounds}) and the simplex symmetrization $\|e(\xi_1)\prod_{i\ge 2}\hat u(\xi_i)\|_{L^2(T_n(x))}\le\|e\|_{L^2}\|\hat u\|^{n-1}_{L^2}\sqrt{x}/\sqrt{(n-1)!}$ \cite[eq.~(27)]{HV2}, multiplying by $c_0\Phi_N(u)\,n$ and summing,
\begin{equation}\label{eq:D2-bd}
|D_2(x,t)|\le c_0\Phi_N(u)\,\|e\|_{L^2}\,\sqrt{x}\,\Lambda^D_K(\|\hat u\|_{L^2}).
\end{equation}
Pairing against $|w|e^{cx}$ and combining with the bound on the first piece,
\begin{equation}\label{eq:D-bd}
|\langle e^{cx}w,D\rangle|\le c_0\Phi_N(u)\bigl(1+e^{c/2}\Lambda^D_K(\|\hat u\|_{L^2})\bigr)\,\|w\|_c\,\|e\|_c.
\end{equation}
Young's inequality \eqref{eq:young} with $\vartheta=\lambda$ on $a=\sqrt{c_0\Phi_N(u)}(1+e^{c/2}\Lambda^D_K)\|w\|_c$, $b=\sqrt{c_0\Phi_N(u)}\|e\|_c$ gives \eqref{eq:D-Y-corrected}.

\emph{(iii) The controller-truncation residual $R^K_N$.} The unmatched-order pieces of the controller cascade have multilinear structure $\check k_n\cdot\hat f_m\cdot\hat u^{\otimes(n+m-1)}$ at total order $n+m-1>N$. By Cauchy--Schwarz on the inner $T_m(\xi_j)$-simplex and the outer $T_n(x)$-simplex with the simplex symmetrization of \cite{HV2}, using $\|\check k_n(x,\cdot)\|_{L^2(T_n(x))}\le\sqrt{n!}\,D_K^\infty(C_K^\infty)^{n-1}x^{n/2}$ from Lemma~\ref{lem:k-bounds} and $\|\hat f_m\|_\infty\le m!D_f/\rho_f^{m-1}$, the $\sqrt{n!}/\sqrt{(n-1)!}=\sqrt{n}$ and $m!/\sqrt{m!}=\sqrt{m!}$ simplifications yield
\begin{eqnarray}
|R^{K,n,m}_N(x,t)| &\le& \sqrt{n\,m!}\,D_K^\infty(C_K^\infty)^{n-1}\,\frac{D_f}{\rho_f^{m-1}} \nonumber\\
&& \times\,\|\hat u\|^{n+m-1}_{L^2}\,x^{n/2+1/2}. \label{eq:RK-piece-bd}
\end{eqnarray}
Summing $|R^K_N|=\sum_{n+m-1>N}|R^{K,n,m}_N|$, the resulting series is geometric in $C_K^\infty\|\hat u\|_{L^2}$ with leading order $\|\hat u\|^{N+1}_{L^2}$:
\begin{eqnarray}
\|R^K_N(\cdot,t)\|_c &\le& e^{c/2}\,L^{R^K}_N\,\|\hat u(\cdot,t)\|_{L^2}^{N+1} \nonumber\\
&& \times\frac{1}{1-C_K^\infty\|\hat u\|_{L^2}}, \label{eq:RK-state}
\end{eqnarray}
on $\|\hat u\|_{L^2}<1/C_K^\infty$, with $L^{R^K}_N$ depending on $D_K^\infty,C_K^\infty,D_f,\rho_f$. Young's inequality \eqref{eq:young} with $\vartheta=c/4$ on $a=\|w\|_c$, $b=\|R^K_N\|_c$ gives \eqref{eq:RK-Y}.

\emph{(iv) The plant--observer mismatch $R^u_N$.} By the multilinear telescoping identity \cite[eq.~(37)]{HV2}, $R^u_N=R^{u,1}_N+\sum_{n\ge 2}R^{u,n}_N$ where
\begin{eqnarray}
R^{u,n}_N(x,t) &:=& \int_{T_n(x)}\hat f_n(x,\xi,t)\sum_{j=1}^n\!\prod_{i<j}u(\xi_i,t) \nonumber\\
&& \times\,e(\xi_j,t)\cdot\prod_{i>j}\hat u(\xi_i,t)\,d\xi. \label{eq:Run-def}
\end{eqnarray}
For $n\ge 2$, by Cauchy--Schwarz on $T_n(x)$ with the simplex symmetrization of \cite{HV2}, bounding $|\hat f_n|\le n!D_f/\rho_f^{n-1}$ and counting the $\sum_j$ as an $n$-fold count, and pairing against $|w|e^{cx}$,
\begin{equation}\label{eq:Ru-ge2}
\Bigl|\Bigl\langle e^{cx}w,\sum_{n\ge 2}R^{u,n}_N\Bigr\rangle\Bigr|\le e^{c/2}\,\Lambda^{u,\ge 2}_N(\bar u)\,\|w\|_c\,\|e\|_c.
\end{equation}
Young's inequality \eqref{eq:young} with $\vartheta=\lambda c/4$ gives \eqref{eq:Ru-ge2-Y}. For $n=1$ the mismatch is genuinely bilinear and linear in the observer error $e$,
\begin{equation}\label{eq:Ru1-explicit}
R^{u,1}_N(x,t)=\int_0^x\hat f_1(x,\xi,t)\,e(\xi,t)\,d\xi,
\end{equation}
with $\|\hat f_1\|_\infty\le D_f$; by Cauchy--Schwarz $|\langle e^{cx}w,R^{u,1}_N\rangle|\le D_f e^{c/2}\|w\|_c\|e\|_c$, and Young's inequality \eqref{eq:young} with $\vartheta=\delta$ gives \eqref{eq:Ru-1-Y}.

\emph{(v) The identifier residual $R_N[u]$.} The unmodeled plant tail
\begin{equation}\label{eq:RN-explicit}
R_N[u](x,t)=\sum_{n>N}\int_{T_n(x)}f_n(x,\xi)\prod_{i=1}^n u(\xi_i,t)\,d\xi
\end{equation}
in \eqref{eq:Vdot-N} satisfies, by the simplex symmetrization $\int_{T_n(x)}|\prod u(\xi_i)|\,d\xi\le x^{n/2}\|u\|^n_{L^2}/n!$ and Assumption~\ref{ass:factorial},
\begin{equation}\label{eq:RN-bound}
\|R_N[u](\cdot,t)\|_\infty\le L^R_N\,\|u\|_{L^2}^{N+1}\cdot\frac{1}{1-\|u\|_{L^2}/\rho_f},
\end{equation}
on $\|u\|_{L^2}<\rho_f$, with $L^R_N:=D_f/\rho_f^N$. Young's inequality \eqref{eq:young} with $\vartheta=c/4$ on $a=\|e\|_c$, $b=e^{c/2}\|R_N[u]\|_\infty$ gives \eqref{eq:RN-Y}.
\end{proof}

%==============================================================
\section{Stability via $V_{\rm tot}$}\label{sec:proof-Udot}
%==============================================================

The joint Lyapunov function is
\begin{equation}\label{eq:U-Lyap-def}
V_{\rm tot}(t):=V_{\rm targ}(t)+\lambda V_{\rm iden}(t),\qquad V_{\rm targ}(t):=\tfrac12\|w(\cdot,t)\|_c^2.
\end{equation}

\begin{lemma}[Norm equivalence]\label{lem:norm-equiv}
There exist constants $\mu_1,\mu_2,G_2,H_2>0$ depending on $D_f,\rho_f,N,c,\lambda$, with $G_2,H_2$ from the inverse transformation $\hat u=w+\check L_N[w]$ in $L^2$-form, such that
\begin{equation}\label{eq:norm-equiv-bd}
\mu_1\Pi(t)^2\le V_{\rm tot}(t)\le\mu_2\Pi(t)^2,
\end{equation}
hence $\alpha_1(\Pi(t))\le V_{\rm tot}(t)\le\alpha_2(\Pi(t))$ with $\alpha_1(s):=\mu_1 s^2,\,\alpha_2(s):=\mu_2 s^2\in\mathcal K_\infty$.
\end{lemma}

\begin{proof}
\emph{Lower bound $\mu_1$:} bound $\Pi^2$ above by a multiple of $V_{\rm tot}$. With $u=\hat u+e$ and $\|\hat u\|_{L^2}\le G_2\|w\|_{L^2}$,
\begin{eqnarray}
\Pi^2 &\le& 2\|\hat u\|^2+2\|e\|^2+\|\hat u\|^2+\sum_n\|\tilde f_n\|^2 \nonumber\\
&\le& 3G_2^2\|w\|_{L^2}^2+2\|e\|_{L^2}^2+\sum_n\|\tilde f_n\|^2. \label{eq:Pi-up}
\end{eqnarray}
Using $\|w\|_{L^2}^2\le\|w\|_c^2=2V_{\rm targ}$ and $\|e\|_{L^2}^2\le\|e\|_c^2\le 2V_{\rm iden}/\lambda\cdot\lambda=2V_{\rm iden}$ and $\sum_n\|\tilde f_n\|^2\le 2\gamma^{\max}V^{\rm par}_{\rm iden}\le 2\gamma^{\max}V_{\rm iden}$,
\begin{equation}\label{eq:Pi-up2}
\Pi^2\le \mu_1^{-1}\,V_{\rm tot},\qquad \mu_1:=\bigl[\max(6G_2^2,4/\lambda,2\gamma^{\max}/\lambda)\bigr]^{-1}.
\end{equation}
\emph{Upper bound $\mu_2$:} by the forward-direction transformation $\|w\|_{L^2}\le H_2\|\hat u\|_{L^2}$, and $\|w\|_c\le e^{c/2}\|w\|_{L^2}$,
\begin{equation}\label{eq:U-up}
V_{\rm tot}=\tfrac12\|w\|_c^2+\tfrac{\lambda}{2}\|e\|_c^2+\tfrac{\lambda}{2}\sum_n\|\tilde f_n\|^2/\gamma_n\le \mu_2\Pi^2,
\end{equation}
with $\mu_2=\max(e^c H_2^2/2,\lambda e^c/2,\lambda/(2\gamma^{\min}))$.
\end{proof}

\begin{proof}[Proof of Theorem~\ref{thm:stability}(a), (c), (d)]
Combining $\dot V_{\rm targ}$ from \eqref{eq:Wdot-expand} with bounds \eqref{eq:Omega-Y}, \eqref{eq:D-Y-corrected}, \eqref{eq:RK-Y}, \eqref{eq:Ru-ge2-Y}, \eqref{eq:Ru-1-Y} (Lemma~\ref{lem:dominance}), and $\dot V_{\rm iden}$ from Lemma~\ref{lem:Vdot-N} with bound \eqref{eq:RN-Y}, into $\dot V_{\rm tot}=\dot V_{\rm targ}+\lambda\dot V_{\rm iden}$, and using $\lambda\ge 16 D_f^2 e^c/c^2$, $\delta=\lambda c/4$, and $c\le 4c_0$, each of the three dissipation terms in $\dot V_{\rm tot}$ is written as the sum of the amounts absorbing individual cross-terms plus a nonnegative remainder,
\begin{eqnarray}
\tfrac{c}{2} &=& \underbrace{\tfrac{c}{8}}_{R^K_N}+\underbrace{\tfrac{c}{8}}_{R^{u,1}_N}+\tfrac{c}{4}, \label{eq:dissip-w}\\
\tfrac{\lambda c}{2} &=& \underbrace{\tfrac{\lambda c}{8}}_{R^{u,1}_N}+\underbrace{\tfrac{\lambda c}{8}}_{R^{u,\ge 2}_N}+\underbrace{\tfrac{\lambda c}{8}}_{R_N[u]}+\tfrac{\lambda c}{8}, \label{eq:dissip-e}\\
\lambda c_0 &=& \underbrace{\tfrac{\lambda c}{8}}_{\Omega}+\underbrace{\tfrac{\lambda c_0}{2}}_{D}+\underbrace{\bigl(\tfrac{\lambda c_0}{2}-\tfrac{\lambda c}{8}\bigr)}_{\ge\,0}. \label{eq:dissip-ePhi}
\end{eqnarray}
The coefficient $c/8$ charged to $R^{u,1}_N$ equals $D_f^2 e^c/(2\delta)$ under $\delta=\lambda c/4$ and $\lambda\ge 16 D_f^2 e^c/c^2$; the leftover on the third line is nonnegative because $c\le 4c_0$. The leftovers $c/4$ and $\lambda c/8$ on the first two lines must in turn dominate the higher-order multipliers collected next.
Collecting the higher-order multipliers,
\begin{eqnarray}
\Theta_w &:=& \frac{2(L^\Omega)^2}{\lambda c}+\frac{c_0(1+e^{c/2}\Lambda^D_K(\bar u))^2}{2\lambda}\Phi_\infty(\bar u) \nonumber\\
&& +\frac{2 e^c\Lambda^{u,\ge 2}_N(\bar u)^2}{\lambda c} \nonumber\\
&& +\frac{2e^c(L^{R^K}_N)^2 G_2^{2(N+1)}}{c(1-C_K^\infty\bar u)^2}\|w\|_c^{2N} \nonumber\\
&& +\frac{2\lambda e^c(L^R_N)^2\,2^{2N+1}G_2^{2(N+1)}}{c(1-\bar u/\rho_f)^2}\|w\|_c^{2N}, \label{eq:Theta-W}\\
\Theta_e &:=& \frac{2\lambda e^c(L^R_N)^2\,2^{2N+1}}{c(1-\bar u/\rho_f)^2}\|e\|_c^{2N}. \label{eq:Theta-E}
\end{eqnarray}
where we have used $\Phi_N(u)\le\Phi_\infty(\bar u)$ on the basin in the second summand of $\Theta_w$. Each summand of $\Theta_w$ and $\Theta_e$ vanishes as $\bar u\to 0$. The result is the joint inequality
\begin{eqnarray}
\dot V_{\rm tot}(t) &\le& -\tfrac12 w^2(0,t)-\tfrac{\lambda}{2}e^2(0,t) \nonumber\\
&& -\bigl[\tfrac{c}{4}-\Theta_w\bigr]\|w\|_c^2-\bigl[\tfrac{\lambda c}{8}-\Theta_e\bigr]\|e\|_c^2 \nonumber\\
&& -\bigl[\tfrac{\lambda c_0}{2}-\tfrac{\lambda c}{8}\bigr]\|e\|_c^2\Phi_N(u). \label{eq:Udot-final}
\end{eqnarray}

Defining $\Omega_R:=\{(u,\hat u,\tilde f):\Pi\le R\}$, by Lemma~\ref{lem:norm-equiv}, $\Omega_R\supseteq\{V_{\rm tot}\le\mu_1 R^2\}$ and $\Omega_R\subseteq\{V_{\rm tot}\le\mu_2 R^2\}$. Choose $\bar R>0$ small enough that on $\Omega_{\bar R}$, $\Theta_w<c/8$ and $\Theta_e<\lambda c/16$ (achievable since $\Theta_w,\Theta_e$ vanish as $R\to 0$). Inside $\Omega_{\bar R}$, \eqref{eq:Udot-final} reduces to
\begin{eqnarray}
\dot V_{\rm tot} &\le& -\tfrac{c}{8}\|w\|_c^2-\tfrac{\lambda c}{16}\|e\|_c^2 \nonumber\\
&& -\bigl(\tfrac{\lambda c_0}{2}-\tfrac{\lambda c}{8}\bigr)\|e\|_c^2\Phi_N(u)\le 0, \label{eq:Udot-clean}
\end{eqnarray}
so the sublevel set $\{V_{\rm tot}\le\mu_1\bar R^2\}$ is forward-invariant. Set $\bar\Pi:=\bar R\sqrt{\mu_1/\mu_2}$. For $\Pi(0)\le\bar\Pi$,
\begin{equation}\label{eq:U0-bound}
V_{\rm tot}(0)\le\mu_2\Pi(0)^2\le\mu_2\bar\Pi^2=\mu_1\bar R^2,
\end{equation}
so the trajectory enters $\{V_{\rm tot}\le\mu_1\bar R^2\}\subseteq\Omega_{\bar R}$, where forward invariance and $\dot V_{\rm tot}\le 0$ apply. For any $t\ge 0$, $V_{\rm tot}(t)\le V_{\rm tot}(0)$, hence by Lemma~\ref{lem:norm-equiv},
\begin{equation}\label{eq:Pi-stability}
\mu_1\Pi(t)^2\le V_{\rm tot}(t)\le V_{\rm tot}(0)\le\mu_2\Pi(0)^2,
\end{equation}
giving $\Pi(t)\le\mu_0\,\Pi(0)$ with $\mu_0:=\sqrt{\mu_2/\mu_1}$. This is Theorem~\ref{thm:stability}(a).

Bound \eqref{eq:thm-fhat-bound} is from Lemma~\ref{lem:Vdot-N}(P1); \eqref{eq:thm-kcheck-bound} from Lemma~\ref{lem:k-bounds}. For \eqref{eq:thm-U-bound}, by Cauchy--Schwarz on $T_n(1)$ at $x=1$ using the $L^2$-slice bound $\|\check k_n(1,\cdot)\|_{L^2(T_n(1))}\le \sqrt{n!}\,D_K^\infty(C_K^\infty)^{n-1}$ from Lemma~\ref{lem:k-bounds} and the simplex symmetrization $\|\prod\hat u\|_{L^2(T_n(1))}\le\|\hat u\|^n_{L^2}/\sqrt{n!}$ \cite{HV2}, the $\sqrt{n!}$ factors cancel:
\begin{eqnarray}
|U(t)| &\le& \sum_{n=1}^N\|\check k_n(1,\cdot)\|_{L^2}\cdot\|\textstyle\prod\hat u\|_{L^2} \nonumber\\
&\le& \sum_{n=1}^N D_K^\infty(C_K^\infty\,\|\hat u\|_{L^2})^{n-1}\,\|\hat u\|_{L^2}, \label{eq:U-bound-derivation}
\end{eqnarray}
which is \eqref{eq:thm-U-bound} on the basin, where $\|\hat u\|_{L^2}\le\Pi(t)\le\mu_0\,\Pi(0)<1/C_K^\infty$. This establishes (c) and (d). Part (b) is proved in Section~\ref{sec:proof-Barbalat}.
\end{proof}

%==============================================================
\section{Barbalat Regulation and Completing Proof of Theorem~\ref{thm:stability}}\label{sec:proof-Barbalat}
%==============================================================

This section establishes Theorem~\ref{thm:stability}(b). We first extract $L^1_t\cap L^\infty_t$ regularity for the energy norms (Step 1), then prove uniform continuity of $\|w\|_c^2$ and $\|e\|_c^2$ in $t$ via an explicit bound on their time derivatives (Steps 2--3), and finally apply Barbalat's lemma and cascade the conclusion to the plant state, observer state, and input (Steps 4--5). The computations below are carried out formally; for the smoothed-projection version of the closed loop they are rigorous and the estimates pass to the limit.

\paragraph{Step 1: $L^1_t\cap L^\infty_t$ regularity.}
From \eqref{eq:Udot-clean} and $V_{\rm tot}(t)\ge 0$, integration in time yields
\begin{eqnarray}
\int_0^\infty\|w(\cdot,t)\|_c^2\,dt &\le& \tfrac{8}{c}\,V_{\rm tot}(0)<\infty, \nonumber\\
\int_0^\infty\|e(\cdot,t)\|_c^2\,dt &\le& \tfrac{16}{\lambda c}\,V_{\rm tot}(0)<\infty, \label{eq:we-L2t}
\end{eqnarray}
so $\|w\|_c^2,\|e\|_c^2\in L^1_t$. Combined with $\|w\|_c,\|e\|_c\in L^\infty_t$ (from $V_{\rm tot}(t)\le V_{\rm tot}(0)$ and Lemma~\ref{lem:norm-equiv}), both are in $L^1_t\cap L^\infty_t$.

\paragraph{Step 2: time derivative of $\|w\|_c^2$.}
Compute
\begin{equation}\label{eq:ddt-w-c2}
\tfrac{d}{dt}\|w\|_c^2=2\int_0^1\!e^{cx}w\,w_t\,dx,
\end{equation}
substitute $w_t=w_x+(I-\check K_N^*)R^u_N+c_0\Phi_N(u)(I-\check K_N^*)e-\Omega+R^K_N$ from \eqref{eq:tgt1-clean}, and integrate by parts in $x$ on the $w_x$-piece using $w(1,t)=0$:
\begin{eqnarray}
\tfrac{d}{dt}\|w\|_c^2 &=& -w^2(0,t)-c\|w\|_c^2 \nonumber\\
&& +2\langle e^{cx}w,(I-\check K_N^*)R^u_N\rangle+2\langle e^{cx}w,D\rangle \nonumber\\
&& -2\langle e^{cx}w,\Omega\rangle+2\langle e^{cx}w,R^K_N\rangle. \label{eq:ddt-w-c2-expanded}
\end{eqnarray}

\paragraph{Step 3: uniform $L^\infty_t$ bound on $\frac{d}{dt}\|w\|_c^2$.}
On the forward-invariant basin $\{\Pi\le\bar\Pi\}$ from Theorem~\ref{thm:stability}(a), all state quantities entering the right-hand side of \eqref{eq:ddt-w-c2-expanded} are bounded uniformly in $t$:
\begin{eqnarray}
\|w\|_c &\le& \sqrt{2\,\mu_2}\,\bar\Pi,\quad \|e\|_c\le\sqrt{2\,\mu_2/\lambda}\,\bar\Pi, \nonumber\\
\|\hat u\|_{L^2},\|u\|_{L^2} &\le& \bar\Pi,\quad \Phi_N(u)\le\Phi_\infty(\bar\Pi), \label{eq:basin-bounds}
\end{eqnarray}
where $\Phi_\infty$ is an $N$-uniform majorant for $\Phi_N$. Each cross-term in \eqref{eq:ddt-w-c2-expanded} is bounded as follows.

\emph{(i) The $R^u_N$ term.} By Lemma~\ref{lem:dominance} (the bilinear and HOT bounds on $R^u_N$) and the simplex-tensor structure of $(I-\check K_N^*)R^u_N$,
\begin{eqnarray}
\lefteqn{|\langle e^{cx}w,(I-\check K_N^*)R^u_N\rangle|} \nonumber\\
&\le& \bigl(D_f e^{c/2}+e^{c/2}\Lambda^{u,\ge 2}_N(\bar\Pi)\bigr) \nonumber\\
&& \times\bigl(1+e^{c/2}\Lambda^D_K(\bar\Pi)\bigr)\|w\|_c\|e\|_c\le M_1, \label{eq:ddt-Ru-bd}
\end{eqnarray}
with $M_1=M_1(\bar\Pi,D_f,N,c)$ a constant depending only on the basin radius and design.

\emph{(ii) The damping term $D$.} From \eqref{eq:D-bd},
\begin{eqnarray}
|\langle e^{cx}w,D\rangle| &\le& c_0\Phi_\infty(\bar\Pi)\bigl(1+e^{c/2}\Lambda^D_K(\bar\Pi)\bigr)\|w\|_c\|e\|_c \nonumber\\
&\le& M_2. \label{eq:ddt-D-bd}
\end{eqnarray}

\emph{(iii) The drift $\Omega$.} From \eqref{eq:Omega-x-bd} paired against $|w|e^{cx}$ and the parameter-rate bound \eqref{eq:P4-N},
\begin{eqnarray}
|\langle e^{cx}w,\Omega\rangle| &\le& e^{c/2}\sqrt{\gamma^{\max}e^c}\,\Lambda^\Omega_N(\bar\Pi)\,\|w\|_c\|e\|_c\sqrt{\Phi_\infty(\bar\Pi)} \nonumber\\
&\le& M_3. \label{eq:ddt-Omega-bd}
\end{eqnarray}

\emph{(iv) The truncation residual $R^K_N$.} From \eqref{eq:RK-state},
\begin{equation}\label{eq:ddt-RK-bd}
|\langle e^{cx}w,R^K_N\rangle|\le\|w\|_c\cdot\frac{e^{c/2}L^{R^K}_N\bar\Pi^{N+1}}{1-C_K^\infty\bar\Pi}\le M_4.
\end{equation}

Combining \eqref{eq:basin-bounds}--\eqref{eq:ddt-RK-bd} into \eqref{eq:ddt-w-c2-expanded},
\begin{eqnarray}
\Bigl|\tfrac{d}{dt}\|w(\cdot,t)\|_c^2\Bigr| &\le& c\cdot 2\mu_2\bar\Pi^2+2(M_1+M_2+M_3+M_4) \nonumber\\
&=:& M_w<\infty, \label{eq:ddt-w-Linf}
\end{eqnarray}
uniformly in $t$ on the basin. (We dropped the negative boundary term $-w^2(0,t)$, which only improves the bound.) Hence $\frac{d}{dt}\|w\|_c^2\in L^\infty_t$, so $\|w\|_c^2$ is uniformly continuous on $[0,\infty)$.

The same calculation applied to $\|e\|_c^2$ via the closed-loop $e$-dynamics \eqref{eq:e-dynamics} yields $\frac{d}{dt}\|e\|_c^2\in L^\infty_t$ analogously, with constant $M_e=M_e(\bar\Pi,D_f,N,c,\lambda,\gamma^{\max})<\infty$ (the cross-terms involve $\tilde F_N[u,\tilde f]$, $R_N[u]$, and the damping $c_0 e\Phi_N(u)$, all bounded on the basin by the projection bounds on $\hat f_n$ and the basin bounds on the state norms).

\paragraph{Step 4: Barbalat applied to $\|w\|_c^2$ and $\|e\|_c^2$.}
By Step 1, $\|w\|_c^2\in L^1_t$. By Step 3, $\|w\|_c^2$ is uniformly continuous on $[0,\infty)$. Barbalat's lemma gives
\begin{equation}\label{eq:Barbalat-w}
\|w(\cdot,t)\|_c^2\to 0,\qquad\text{hence}\qquad\|w(\cdot,t)\|_c\to 0,\quad\text{as }t\to\infty.
\end{equation}
The same argument gives $\|e(\cdot,t)\|_c\to 0$.

\paragraph{Step 5: cascade to plant, observer, and input.}
By the inverse transformation $\hat u=w+\check L_N[w]$ \eqref{eq:inverse-transform},
\begin{equation}\label{eq:cascade-uhat}
\|\hat u(\cdot,t)\|_{L^2}\le G_2\,\|w(\cdot,t)\|_{L^2}\le G_2\,\|w(\cdot,t)\|_c\to 0.
\end{equation}
By $u=\hat u+e$,
\begin{equation}\label{eq:cascade-u}
\|u(\cdot,t)\|_{L^2}\le\|\hat u(\cdot,t)\|_{L^2}+\|e(\cdot,t)\|_{L^2}\to 0.
\end{equation}
By the input bound \eqref{eq:thm-U-bound},
\begin{equation}\label{eq:cascade-U}
|U(t)|\le\sum_{n=1}^N D_K^\infty(C_K^\infty\,\|\hat u\|_{L^2})^{n-1}\,\|\hat u\|_{L^2}\to 0,
\end{equation}
since each summand goes to zero as $\|\hat u\|_{L^2}\to 0$, and the sum is finite. This establishes \eqref{eq:thm-regulation}.

%==============================================================
\section{Start of Proof of Theorem~\ref{thm:neural}: Lipschitzness of $\mathcal Q_N$}\label{sec:QN-Lip}
%==============================================================

With Theorem~\ref{thm:stability} now fully proven, we turn to the proof of Theorem~\ref{thm:neural}.

The kernel-generation operator $\mathcal Q_N$ of \eqref{eq:QN-def}, which sends the parameter estimate tuple $\hat f=(\hat f_1,\ldots,\hat f_N)$ to the kernel tuple $\check k=(\check k_1,\ldots,\check k_N)$ via the cascade \eqref{eq:k-IE}, is here shown to be Lipschitz continuous from $\prod_n \bar F_n\subset\prod_n L^\infty(T_n(1))$ to $\prod_n L^2(T_n(1))$.

\begin{lemma}[Lipschitz continuity of $\mathcal Q_N$]\label{lem:Q-Lip}
For all $\hat f,\hat f'\in\prod_{n=1}^N \bar F_n$,
\begin{eqnarray}
\lefteqn{\|\mathcal Q_N(\hat f)_n-\mathcal Q_N(\hat f')_n\|_{L^2(T_n(1))}} \nonumber\\
&\le& \ell_n\,\|\hat f-\hat f'\|_\star,\quad n=1,\ldots,N, \label{eq:Q-Lip}
\end{eqnarray}
where the norm on the parameter tuple is
\begin{equation}\label{eq:star-norm}
\|\hat f-\hat f'\|_\star:=\max_{n=1,\ldots,N}\|\hat f_n-\hat f_n'\|_{L^\infty(T_n(1))},
\end{equation}
and the Lipschitz constants $\ell_n>0$ are defined by the recursion
\begin{eqnarray}
\ell_n &:=& (1+M_n e^{M_n})\Big[\tfrac{1}{\sqrt{n!}} \nonumber\\
&& +\sum_{m=2}^n\tbinom{n+1}{m+1}\big(\tfrac{m!D_f}{\rho_f^{m-1}}\ell_{n-m+1} \nonumber\\
&& \quad+\sqrt{(n-m+1)!}D_K^\infty(C_K^\infty)^{n-m}\big) \nonumber\\
&& +\tbinom{n+1}{2}\sqrt{n!}D_K^\infty(C_K^\infty)^{n-1}\Big],\ n\ge 1, \label{eq:lambda-recursion}
\end{eqnarray}
with $M_n:=\binom{n+1}{2}D_f$ as in Lemma~\ref{lem:Sn-inverse}, $D_K^\infty,C_K^\infty$ as in Lemma~\ref{lem:k-bounds}, and $\binom{n+1}{2}\sqrt{n!}D_K^\infty(C_K^\infty)^{n-1}$ the contribution from the self-coupling perturbation $\check{\mathcal S}_n-\check{\mathcal S}_n'$. For $n=1$ the sum $\sum_{m=2}^n$ is empty, giving the base case
\begin{equation}\label{eq:lambda-base}
\ell_1=(1+D_f e^{D_f})(1+D_K^\infty).
\end{equation}
\end{lemma}

\begin{proof}
Cascade induction on $n$. The kernels $\check k_n$ (for $\hat f$) and $\check k_n'$ (for $\hat f'$) satisfy
\begin{equation}\label{eq:k-IE-pair}
(I-\check{\mathcal S}_n)\,\check k_n=\check\Psi_n,\qquad (I-\check{\mathcal S}_n')\,\check k_n'=\check\Psi_n',
\end{equation}
where $\check{\mathcal S}_n=\check{\mathcal S}_n[\hat f_1]$ depends on $\hat f_1$ alone (eq.~\eqref{eq:Sn-def}) and $\check\Psi_n=\check\Psi_n[\hat f_n,\check k_{n-m+1},\hat f_m\,;\,m\ge 2]$ depends on $\hat f_n$ and the lower-order kernels (eq.~\eqref{eq:Psin-def}). Subtracting:
\begin{equation}\label{eq:k-diff-IE}
(I-\check{\mathcal S}_n)(\check k_n-\check k_n')=(\check{\mathcal S}_n-\check{\mathcal S}_n')\check k_n'+(\check\Psi_n-\check\Psi_n').
\end{equation}
By Lemma~\ref{lem:Sn-inverse}, $\|(I-\check{\mathcal S}_n)^{-1}\|_{L^2\to L^2}\le 1+M_n e^{M_n}$. Hence
\begin{eqnarray}
\|\check k_n-\check k_n'\|_{L^2(T_n(1))} &\le& (1+M_n e^{M_n}) \nonumber\\
&& \times\Big[\|(\check{\mathcal S}_n-\check{\mathcal S}_n')\check k_n'\|_{L^2} \nonumber\\
&& \quad+\|\check\Psi_n-\check\Psi_n'\|_{L^2}\Big]. \label{eq:k-diff-L2}
\end{eqnarray}
\emph{The self-coupling perturbation $(\check{\mathcal S}_n-\check{\mathcal S}_n')\check k_n'$.} By bilinearity of $B^1_n$ in its second argument, $(\check{\mathcal S}_n-\check{\mathcal S}_n')[\kappa]=\int_0^{\xi_n}B^1_n[\kappa,\hat f_1-\hat f_1'](\xi^\sigma)d\sigma$. By the same characteristic-coordinate $L^2$-bound as in Lemma~\ref{lem:Sn-inverse} (applied with $k=1$, no factorial gain): for any $\kappa\in L^2(T_n(1))$ and $g\in L^\infty(T_1(1))$,
\begin{eqnarray}
\lefteqn{\Bigl\|\int_0^{\xi_n}B^1_n[\kappa,g](\xi^\sigma)d\sigma\Bigr\|_{L^2(T_n(1))}} \nonumber\\
&\le& \binom{n+1}{2}\|g\|_\infty\,\|\kappa\|_{L^2(T_n(1))}. \label{eq:S-diff-bound}
\end{eqnarray}
Hence $\|(\check{\mathcal S}_n-\check{\mathcal S}_n')\check k_n'\|_{L^2}\le \binom{n+1}{2}\|\hat f_1-\hat f_1'\|_\infty\,\|\check k_n'\|_{L^2}$. By Lemma~\ref{lem:k-bounds} at $x=1$, $\|\check k_n'\|_{L^2(T_n(1))}\le\sqrt{n!}\,D_K^\infty(C_K^\infty)^{n-1}$, so
\begin{equation}\label{eq:S-diff-final}
\|(\check{\mathcal S}_n-\check{\mathcal S}_n')\check k_n'\|_{L^2}\le \binom{n+1}{2}\sqrt{n!}\,D_K^\infty(C_K^\infty)^{n-1}\,\|\hat f-\hat f'\|_\star.
\end{equation}
\emph{The inhomogeneity perturbation $\check\Psi_n-\check\Psi_n'$.} From \eqref{eq:Psin-def},
\begin{eqnarray}
\check\Psi_n-\check\Psi_n' &=& -\int_0^{\xi_n}\!(\hat f_n-\hat f_n')(\xi^\sigma)d\sigma \nonumber\\
&& +\sum_{m=2}^n\!\int_0^{\xi_n}\!\!\big[B^m_n[\check k_{n-m+1},\hat f_m] \nonumber\\
&& -B^m_n[\check k_{n-m+1}',\hat f_m']\big](\xi^\sigma)d\sigma. \label{eq:Psi-diff}
\end{eqnarray}
The first piece, in characteristic coordinates $(\eta,\xi_n)$ where the displacement leaves $\eta$ unchanged and replaces $\xi_n$ by $\sigma$, has $L^2$-norm $\le\|\hat f_n-\hat f_n'\|_\infty\cdot\sqrt{\mathrm{vol}(T_n(1))}=\|\hat f_n-\hat f_n'\|_\infty/\sqrt{n!}\le\|\hat f-\hat f'\|_\star/\sqrt{n!}$.
For the second piece, by bilinearity of $B^m_n$,
\begin{eqnarray}
\lefteqn{B^m_n[\check k_{n-m+1},\hat f_m]-B^m_n[\check k_{n-m+1}',\hat f_m']=} \nonumber\\
&& B^m_n[\check k_{n-m+1}-\check k_{n-m+1}',\hat f_m] \nonumber\\
&& +B^m_n[\check k_{n-m+1}',\hat f_m-\hat f_m']. \label{eq:Bm-bilin-split}
\end{eqnarray}
By \cite[Lemma~A.1]{HV1} extended to $L^2$ ($B^m_n$ produces $\binom{n+1}{m+1}$ pointwise terms each bounded in $L^2$ by Cauchy--Schwarz),
\begin{equation}\label{eq:Bm-L2}
\|B^m_n[\kappa,g]\|_{L^2(T_n(1))}\le \binom{n+1}{m+1}\|g\|_\infty\,\|\kappa\|_{L^2(T_{n-m+1}(1))}.
\end{equation}
Applying \eqref{eq:Bm-L2} to each summand of \eqref{eq:Bm-bilin-split}, with $\|\hat f_m\|_\infty\le m!D_f/\rho_f^{m-1}$ and $\|\check k_{n-m+1}'\|_{L^2}\le\sqrt{(n-m+1)!}D_K^\infty(C_K^\infty)^{n-m}$, and using the inductive hypothesis $\|\check k_{n-m+1}-\check k_{n-m+1}'\|_{L^2}\le\ell_{n-m+1}\|\hat f-\hat f'\|_\star$,
\begin{eqnarray}
\lefteqn{\|B^m_n[\check k_{n-m+1},\hat f_m]-B^m_n[\check k_{n-m+1}',\hat f_m']\|_{L^2}} \nonumber\\
&\le& \binom{n+1}{m+1}\Big[\frac{m!D_f}{\rho_f^{m-1}}\ell_{n-m+1} \nonumber\\
&& +\sqrt{(n-m+1)!}D_K^\infty(C_K^\infty)^{n-m}\Big]\|\hat f-\hat f'\|_\star. \label{eq:Bm-diff-final}
\end{eqnarray}
The outer $\sigma$-integration $\int_0^{\xi_n}\cdots d\sigma$ in $L^2(T_n(1))$, in characteristic coordinates, contributes at most a factor 1 (same argument as in the proof of Lemma~\ref{lem:Sn-inverse}, $k=1$ case). Hence
\begin{eqnarray}
\lefteqn{\|\check\Psi_n-\check\Psi_n'\|_{L^2(T_n(1))}\le\Big[\tfrac{1}{\sqrt{n!}}+\sum_{m=2}^n\tbinom{n+1}{m+1}} \nonumber\\
&& \times\big(\tfrac{m!D_f}{\rho_f^{m-1}}\ell_{n-m+1}+\sqrt{(n-m+1)!}D_K^\infty(C_K^\infty)^{n-m}\big) \nonumber\\
&& \Big]\|\hat f-\hat f'\|_\star. \label{eq:Psi-diff-bound}
\end{eqnarray}
\emph{Combining.} Substituting \eqref{eq:S-diff-final} and \eqref{eq:Psi-diff-bound} into \eqref{eq:k-diff-L2} gives \eqref{eq:Q-Lip} with $\ell_n$ as in \eqref{eq:lambda-recursion}, valid for all $n\ge 1$. For $n=1$ the inductive sum $\sum_{m=2}^n$ is empty, but the self-coupling term $\binom{2}{2}\sqrt{1!}D_K^\infty(C_K^\infty)^0=D_K^\infty$ from $(\check{\mathcal S}_1-\check{\mathcal S}_1')\check k_1'$ is present and combined with the linear-inhomogeneity term $1/\sqrt{1!}=1$, giving the base case \eqref{eq:lambda-base}.
\end{proof}

%==============================================================
\section{Neural-Operator Approximation of $\mathcal Q_N$}\label{sec:NO-approx}
%==============================================================

This section establishes the existence of a neural-operator surrogate $\widehat{\mathcal Q}_N$ for the kernel-generation operator $\mathcal Q_N$ of \eqref{eq:QN-def}, with prescribed accuracy on the Lipschitz-restricted projection box $\prod_n \bar F_n^{\rm Lip}$ from \eqref{eq:Bn-Lip}. The argument follows \cite[\S5]{HV3}: Arzel\`a--Ascoli compactness from the Lipschitz seminorm constraint on the training inputs, combined with continuity of $\mathcal Q_N$ from Lemma~\ref{lem:Q-Lip}, plus the neural-operator universal approximation theorem of \cite{Lu2021Universal}.

\begin{lemma}[Compactness of the training class]\label{lem:compactness}
The set $\prod_{n=1}^N \bar F_n^{\rm Lip}$ is compact in $\prod_{n=1}^N L^\infty(T_n(1))$, and hence (via the continuous embedding $L^\infty(T_n(1))\hookrightarrow L^2(T_n(1))$) also compact in $\prod_n L^2(T_n(1))$.
\end{lemma}

\begin{proof}
Each $\bar F_n^{\rm Lip}$ is a closed, uniformly bounded, equicontinuous subset of $C(T_n(1))$ (uniform bounds from $\|g\|_\infty\le \bar f_n$, equicontinuity from $\mathrm{Lip}(g)\le L_f$). By Arzel\`a--Ascoli, $\bar F_n^{\rm Lip}$ is compact in $C(T_n(1))$, hence in $L^\infty(T_n(1))$ (uniform convergence implies $L^\infty$-convergence on bounded simplices). The finite product is compact in the product topology, and the embedding into $\prod_n L^2(T_n(1))$ on the bounded simplex domain is continuous.
\end{proof}

\begin{theorem}[Neural-operator approximation of $\mathcal Q_N$]\label{thm:NO-approx}
For every $\varepsilon>0$, there exists a neural operator
\begin{equation}\label{eq:Qhat-def}
\widehat{\mathcal Q}_N:\;\prod_{n=1}^N \bar F_n^{\rm Lip}\;\longrightarrow\;\prod_{n=1}^N L^2(T_n(1))
\end{equation}
such that, for $n=1,\ldots,N$ and all $\hat f\in\prod_{n=1}^N \bar F_n^{\rm Lip}$,
\begin{equation}\label{eq:Qhat-eps}
\|\widehat{\mathcal Q}_N(\hat f)_n-\mathcal Q_N(\hat f)_n\|_{L^2(T_n(1))}\le \varepsilon.
\end{equation}
\end{theorem}

\begin{proof}
By Lemma~\ref{lem:Q-Lip}, the kernel-generation operator $\mathcal Q_N:\prod_n \bar F_n^{\rm Lip}\to\prod_n L^2(T_n(1))$ is Lipschitz, hence continuous, where the product spaces are endowed with the $\ell^2$-type product norms $\|(g_1,\ldots,g_N)\|_{\prod L^\infty}^2:=\sum_n\|g_n\|_{L^\infty(T_n(1))}^2$ and $\|(h_1,\ldots,h_N)\|_{\prod L^2}^2:=\sum_n\|h_n\|_{L^2(T_n(1))}^2$. (Each component norm is bounded above by the product $\ell^2$-norm, so a bound on the latter immediately yields a componentwise bound.) The codomain $\prod_n L^2(T_n(1))$ is a separable Banach space. Although $\prod_n L^\infty(T_n(1))$ is not separable in general, the training class $\prod_n \bar F_n^{\rm Lip}$ embeds into the separable Banach space $\prod_n C(T_n(1))$ via the natural inclusion $\bar F_n^{\rm Lip}\subset C(T_n(1))$ (uniform-continuity is automatic for Lipschitz functions on a bounded simplex). By Lemma~\ref{lem:compactness}, $\prod_n \bar F_n^{\rm Lip}$ is compact in $\prod_n C(T_n(1))$ (uniform-norm topology coincides with $L^\infty$-norm topology on $C$). The neural-operator universal approximation theorem of \cite{Lu2021Universal}, in the form for continuous maps between separable Banach spaces restricted to a compact domain subset (an analogous statement is given in \cite{HV3} for the closely related operator $\mathcal U_N$), yields, for any $\varepsilon>0$, a neural operator $\widehat{\mathcal Q}_N$ approximating $\mathcal Q_N$ uniformly on $\prod_n \bar F_n^{\rm Lip}$ in the $\prod_n L^2$-norm (the $\ell^2$ product norm just specified) with accuracy $\varepsilon$. Since each component $L^2$-norm is bounded above by the product $\ell^2$-norm, \eqref{eq:Qhat-eps} follows componentwise.
\end{proof}

%==============================================================
\section{Completing the Proof of Theorem~\ref{thm:neural}}\label{sec:neural-closed-loop}
%==============================================================

In Section~\ref{sec:NO-approx} (Theorem~\ref{thm:NO-approx}), we have established that for every $\varepsilon>0$ there exists a neural-operator surrogate $\widehat{\mathcal Q}_N$ of the exact kernel-generation operator $\mathcal Q_N$, with $L^2(T_n(1))$-accuracy $\varepsilon$ on the Lipschitz-restricted projection box $\prod_n \bar F_n^{\rm Lip}$. In this section we fix such a $\widehat{\mathcal Q}_N$ for an $\varepsilon$ to be chosen sufficiently small below, replace the exact kernels $\check k_n=\mathcal Q_N(\hat f)_n$ in the implemented controller \eqref{eq:U-def} by the neural-operator surrogate kernels $\hat k_n=\widehat{\mathcal Q}_N(\hat f)_n$, and establish that conclusions (a)--(b) of Theorem~\ref{thm:stability} continue to hold on a (possibly smaller) basin $\{\Pi\le\bar\Pi'\}$.

The analysis is a passive-identifier design: the backstepping transformation in the proof uses the \emph{exact} cascade kernel $\check k_n$ as an analytical object, while the implemented controller uses the learned $\hat k_n$. The neural-approximation error then enters the closed-loop analysis through a perturbed boundary condition $w(1,t)=\Gamma(t)$ rather than through the interior equation, and no approximation of the time-derivative $\partial_t\hat k_n$ of the learned kernel is required: the time-derivative drift $\Omega$ in the target system involves only $\partial_t\check k_n$, controlled exactly as in Theorem~\ref{thm:stability} via Lemma~\ref{lem:Sn-inverse} and the update law \eqref{eq:upd1-N}. The combination of the passive identifier with neural-operator implementation is the strategy of \cite[\S7]{LBSK} for the linear-recirculation benchmark; the extension here is to the Volterra-nonlinear cascade with $f_1\not\equiv 0$.

\begin{proof}[Proof of Theorem~\ref{thm:neural}]
\medskip\noindent\textbf{Step 1: Neural surrogate controller and backstepping transformation.}
The implemented (neural surrogate) controller is
\begin{eqnarray}
U(t) &:=& \hat K_N[\hat u](1,t) \nonumber\\
&:=& \sum_{n=1}^N\int_{T_n(1)}\hat k_n(1,\xi,t)\prod_{i=1}^n\hat u(\xi_i,t)\,d\xi, \nonumber\\
\hat k_n(t) &:=& \widehat{\mathcal Q}_N\bigl(\hat f(t)\bigr)_n. \label{eq:Uhat-def}
\end{eqnarray}
In contrast to the exact controller \eqref{eq:U-def}, no Volterra cascade is solved online: $\hat k_n(t)$ is obtained by a single forward evaluation of the offline-trained neural operator $\widehat{\mathcal Q}_N$ at the current parameter estimate $\hat f(t)$. The backstepping transformation, however, is defined with the \emph{exact} kernels $\check k_n=\mathcal Q_N(\hat f)_n$ (the solution of the cascade \eqref{eq:k-IE} at the current $\hat f(t)$):
\begin{eqnarray}
w(x,t) &:=& \hat u(x,t)-\check K_N[\hat u](x,t), \nonumber\\
\check K_N[\hat u](x,t) &:=& \sum_{n=1}^N\int_{T_n(x)}\check k_n(x,\xi,t)\prod_{i=1}^n\hat u(\xi_i,t)\,d\xi. \label{eq:w-def-passive}
\end{eqnarray}
The exact $\check k_n$ enters the proof only as an analytical comparison object through the kernel-approximation error
\begin{eqnarray}
\Delta_n(x,\xi,t) &:=& \check k_n(x,\xi,t)-\hat k_n(x,\xi,t), \nonumber\\
\|\Delta_n(\cdot,\cdot,t)\|_{L^2(T_n(1))} &\le& \varepsilon, \label{eq:Delta-def}
\end{eqnarray}
the latter inequality holding for all $t\ge 0$ under the regularity hypothesis $\hat f(\cdot,\cdot,t)\in\prod_n \bar F_n^{\rm Lip}$ from \eqref{eq:Lip-traj-hyp}.

\medskip\noindent\textbf{Step 2: Target equation --- interior dynamics unchanged, boundary perturbed.}
Because the transformation \eqref{eq:w-def-passive} uses the exact $\check k_n$ --- the very kernels that satisfy the matching identity \eqref{eq:matching-PDE} --- the target equation derivation in Section~\ref{sec:proof-target} carries over verbatim. The target equation is
\begin{eqnarray}
w_t &=& w_x+(I-\check K_N^*)R^u_N \nonumber\\
&& +\,c_0\Phi_N(u)\,(I-\check K_N^*)e-\Omega+R^K_N, \label{eq:tgt1-neural}\\
w(1,t) &=& \Gamma(t), \label{eq:tgt2-neural}
\end{eqnarray}
identical in the interior to \eqref{eq:tgt1-clean}, with the same four perturbations $R^u_N,D,\Omega,R^K_N$ defined in Section~\ref{sec:proof-target}, and with $\Omega=\sum_n\int_{T_n(x)}\partial_t\check k_n\prod_i\hat u\,d\xi$ involving the \emph{exact} kernel's time-derivative (controllable through the cascade differentiation as in the exact-kernel proof, with no neural-operator hypothesis required). The only difference from the exact-kernel case is the boundary condition: $w(1,t)=0$ no longer holds. Instead,
\begin{eqnarray}
\Gamma(t) &:=& w(1,t)=\hat u(1,t)-\check K_N[\hat u](1,t) \nonumber\\
&=& U(t)-\check K_N[\hat u](1,t) \nonumber\\
&=& -\sum_{n=1}^N\int_{T_n(1)}\Delta_n(1,\xi,t)\prod_{i=1}^n\hat u(\xi_i,t)\,d\xi, \label{eq:Gamma-def}
\end{eqnarray}
where the last equality uses $U(t)=\hat K_N[\hat u](1,t)$ from \eqref{eq:Uhat-def} together with $\check K_N-\hat K_N=\sum_n\int\Delta_n\prod\hat u$. The boundary residual $\Gamma(t)$ is the only manifestation of the neural-approximation error in the closed loop.

\medskip\noindent\textbf{Step 3: Boundary residual bound.}
By Cauchy--Schwarz on $T_n(1)$ with the simplex bridge of \cite{HV2} $\|\prod_i\hat u(\xi_i)\|_{L^2(T_n(1))}=\|\hat u\|^n_{L^2(0,1)}/\sqrt{n!}$,
\begin{eqnarray}
|\Gamma(t)| &\le& \sum_{n=1}^N\|\Delta_n(1,\cdot,t)\|_{L^2(T_n(1))}\cdot\frac{\|\hat u(\cdot,t)\|^n_{L^2}}{\sqrt{n!}} \nonumber\\
&\le& \varepsilon\sum_{n=1}^N\frac{\|\hat u(\cdot,t)\|^n_{L^2}}{\sqrt{n!}}=:\varepsilon\,\Lambda^\Gamma_N(\|\hat u\|_{L^2}), \label{eq:Gamma-bound}
\end{eqnarray}
where $\Lambda^\Gamma_N(r):=\sum_{n=1}^N r^n/\sqrt{n!}$. Because the series starts at $n=1$,
\begin{equation}\label{eq:Gamma-vanishing}
\Lambda^\Gamma_N(0)=0,\qquad \Lambda^\Gamma_N(r)=r+O(r^2)\quad\text{as }r\to 0,
\end{equation}
so the boundary residual is state-vanishing---it injects no constant term, which is what carries regulation through to the neural closed loop. Squaring \eqref{eq:Gamma-bound}, using the inverse transformation $\|\hat u\|_{L^2}\le G_2\|w\|_c$ on the basin, and the leading-order bound $\Lambda^\Gamma_N(r)^2\le H^\Gamma_N(\bar u)\,r^2$ with the polynomial
\begin{equation}\label{eq:HGamma-def}
H^\Gamma_N(\bar u):=\Bigl(\sum_{n=1}^N \bar u^{\,n-1}/\sqrt{n!}\Bigr)^2
\end{equation}
bounded on the basin,
\begin{equation}\label{eq:Gamma-sq-bound}
\Gamma^2(t)\le \varepsilon^2 H^\Gamma_N(\bar u)\|\hat u\|^2_{L^2}\le \varepsilon^2 G_2^2 H^\Gamma_N(\bar u)\|w(\cdot,t)\|_c^2.
\end{equation}

\medskip\noindent\textbf{Step 4: Modified Lyapunov closure.}
The Lyapunov-target component $V_{\rm targ}=\tfrac12\|w\|_c^2$ now satisfies, via integration by parts in \eqref{eq:tgt1-neural} taking into account $w(1,t)=\Gamma(t)\ne 0$,
\begin{eqnarray}
\dot V_{\rm targ} &=& \tfrac12 e^c\Gamma^2(t)-\tfrac12 w^2(0,t)-\tfrac{c}{2}\|w\|_c^2 \nonumber\\
&& +(\text{interior cross-terms from $R^u_N,D,\Omega,R^K_N$}), \label{eq:Vtarg-dot-passive}
\end{eqnarray}
where the interior cross-terms are exactly the ones bounded in Section~\ref{sec:proof-dominance}. The new term $\tfrac12 e^c\Gamma^2(t)$ is absorbed into the dissipation via \eqref{eq:Gamma-sq-bound}:
\begin{equation}\label{eq:Gamma-absorption}
\tfrac12 e^c\Gamma^2(t)\le \Theta_w^\Gamma\|w(\cdot,t)\|_c^2,\qquad \Theta_w^\Gamma:=\tfrac12 e^c G_2^2 H^\Gamma_N(\bar u)\varepsilon^2.
\end{equation}
The Lyapunov inequality \eqref{eq:Udot-final} thus becomes
\begin{eqnarray}
\dot V_{\rm tot} &\le& -\tfrac12 w^2(0,t)-\tfrac{\lambda}{2}e^2(0,t) \nonumber\\
&& -\bigl[\tfrac{c}{4}-\Theta_w-\Theta_w^\Gamma\bigr]\|w\|_c^2-\bigl[\tfrac{\lambda c}{8}-\Theta_e\bigr]\|e\|_c^2 \nonumber\\
&& -\bigl[\tfrac{\lambda c_0}{2}-\tfrac{\lambda c}{8}\bigr]\|e\|_c^2\Phi_N(u), \label{eq:Udot-final-neural}
\end{eqnarray}
with the same $\Theta_w,\Theta_e$ as in Theorem~\ref{thm:stability} and the additional $\Theta_w^\Gamma=O(\varepsilon^2)$ from the boundary residual.

\medskip\noindent\textbf{Step 5: Stability and asymptotic regulation.}
The Lyapunov inequality \eqref{eq:Udot-final-neural} has the same structure as the exact-kernel inequality \eqref{eq:Udot-final}: dissipation in $\|w\|_c^2,\|e\|_c^2$ with state-dependent multipliers $\Theta_w+\Theta_w^\Gamma$ and $\Theta_e$, and no constant residual. The smallness conditions $\Theta_w+\Theta_w^\Gamma<c/8$, $\Theta_e<\lambda c/16$ define a basin radius $\bar\Pi'$ via the analysis of Section~\ref{sec:proof-Udot}, with $\bar\Pi'\le\bar\Pi$ and $\bar\Pi'\to\bar\Pi$ as $\varepsilon\to 0$. Lyapunov stability (Theorem~\ref{thm:neural}(a)) follows from \eqref{eq:Udot-final-neural} via Lemma~\ref{lem:norm-equiv}. Asymptotic regulation (Theorem~\ref{thm:neural}(b)) follows by the Barbalat argument of Section~\ref{sec:proof-Barbalat}: $\|w(\cdot,t)\|_c,\|e(\cdot,t)\|_c\to 0$, and the cascade to $\|u\|_{L^2}\to 0$, $\|\hat u\|_{L^2}\to 0$ proceeds as in Section~\ref{sec:proof-Barbalat}.

\medskip\noindent\textbf{Step 6: Kernel and input bounds.}
The kernel and input bounds asserted in Theorem~\ref{thm:neural}(c) are, for $n=1,\ldots,N$ and $t\ge 0$,
\begin{eqnarray}
\|\hat k_n(\cdot,\cdot,t)\|_{L^2(T_n(1))} &\le& \sqrt{n!}\,D_K^\infty(C_K^\infty)^{n-1}+\varepsilon, \label{eq:khat-L2-bound}\\
|U(t)| &\le& \sum_{n=1}^N\bigl(\sqrt{n!}\,D_K^\infty(C_K^\infty)^{n-1}+\varepsilon\bigr) \nonumber\\
&& \times\frac{\|\hat u(\cdot,t)\|^n_{L^2}}{\sqrt{n!}}<\infty. \label{eq:Uhat-bound}
\end{eqnarray}
The bound \eqref{eq:khat-L2-bound} combines the exact-kernel $L^2$-bound $\|\check k_n\|_{L^2(T_n(1))}\le\sqrt{n!}D_K^\infty(C_K^\infty)^{n-1}$ from Lemma~\ref{lem:k-bounds} (at $x=1$) with the approximation accuracy $\|\hat k_n-\check k_n\|_{L^2(T_n(1))}\le\varepsilon$ from \eqref{eq:Qhat-eps}, via the triangle inequality. The bound \eqref{eq:Uhat-bound} comes from Cauchy--Schwarz on $T_n(1)$ in \eqref{eq:Uhat-def} with the simplex bridge of \cite{HV2} and \eqref{eq:khat-L2-bound}. Asymptotic regulation $|U(t)|\to 0$ follows from \eqref{eq:Uhat-bound} and $\|\hat u(\cdot,t)\|_{L^2}\to 0$.

At $\varepsilon=0$ (i.e., when $\widehat{\mathcal Q}_N$ is replaced by the exact $\mathcal Q_N$), $\Delta_n=0$, $\Gamma=0$, $\Theta_w^\Gamma=0$, $\bar\Pi'=\bar\Pi$, and Theorem~\ref{thm:stability} is recovered exactly.
\end{proof}

%==============================================================
\section{Numerical Illustration}\label{sec:sims}
%==============================================================

\begin{figure*}[t]
\centering
\includegraphics[width=\textwidth]{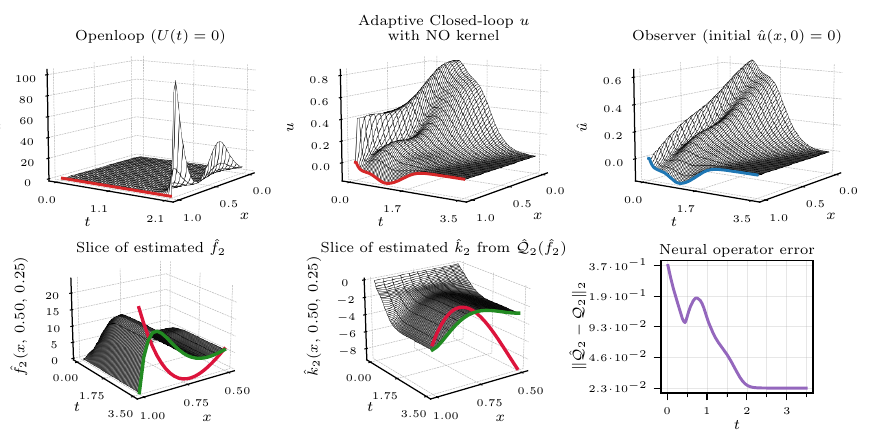}
\caption{Adaptive closed loop at $N=2$ with the neural-operator kernel. \emph{Top:} the open loop ($U\equiv0$) blows up in finite time (left); under the surrogate-generated boundary input the plant state $u$ is regulated to zero (center); the observer $\hat u$, started from $\hat u(\cdot,0)\equiv0$, converges to it (right). \emph{Bottom:} the identifier's coefficient estimate $\hat f_2$ (left) and the kernel $\hat k_2=\widehat{\mathcal Q}_2(\hat f_2)$ the operator returns from it (center), sliced through the simplex at $(\xi_1,\xi_2)=(0.5,0.25)$; and the in-loop operator error $\|\widehat{\mathcal Q}_2-\mathcal Q_2\|_2$ (right), which settles near $2\times10^{-2}$ as the estimate stabilizes. The estimates settle early and then freeze, well before the states have finished decaying, and do not reach the true profiles.}
\label{fig:sims}
\vspace{-1em}
\end{figure*}

We close with an $N=2$ adaptive closed loop in which the controller kernel is supplied, at every step, by the neural operator rather than by an online cascade solve. The plant is a quadratic Volterra PDE from the family of \cite{HV3},
\begin{subequations}
\begin{align}
    u_t(x, t) =&\, u_x(x, t) + F_2[u](x, t)\,, \\
    F_2[u](x, t) =&\, \int_0^x \int_0^{\xi_1} f_2(x, \xi_1, \xi_2)u(\xi_1, t)u(\xi_2, t)\, d\xi_2\, d\xi_1\,, \\
    f_2(x, \xi_1, \xi_2) =&\, A^3 \rho(x, \gamma)\rho(\xi_1, \gamma)\rho(\xi_2, \gamma)\,,
\end{align}
\end{subequations}
where $\rho(z, \gamma) = \cos^2(\gamma\arccos z)$. Its coefficient $f_2$ is unknown and estimated online by the passive identifier of Section~\ref{sec:design} (first-order Euler), with the observer initialized at $\hat u(\cdot,0)\equiv0$, as Theorem~\ref{thm:stability} permits. The run of Figure~\ref{fig:sims} takes $A=2$, $\gamma=5$, the ramp initial condition $u_0(x)=0.4\min(x/0.05,1)$, $\hat u_0\equiv0$, and the initial estimate $\hat f_2(x,\xi_1,\xi_2,0)\equiv 0$. The design uses damping $c_0=3$ in \eqref{eq:obs1-N}, weight $c=2.5$ (with $c\le 4c_0=12$) in \eqref{eq:upd2-N}, and adaptation gains $\gamma_1=0$ (the plant is purely quadratic) and $\gamma_2=60$ in \eqref{eq:Phi-N}; the plant and observer are integrated by first-order upwind on $\Delta x=0.025$, $\Delta t=0.0225$.

\paragraph*{Training the operator} The surrogate $\widehat{\mathcal Q}_2$ of $\mathcal Q_2:\hat f\mapsto\check k$ is trained once, offline, in the manner of \cite{LBSK}, on supervised pairs $(f_2,\check k_2)$: the input is a $41^3$ pointwise evaluation of $f_2$ on the simplex and the output is the boundary gain $\check k_2(1,\cdot,\cdot)$, following the gain-only approach of \cite{VK241}. Since the projection in \eqref{eq:upd1-N} confines $\hat f(t)$ to the Lipschitz box $\bar F_2^{\rm Lip}$ for all $t$, the operator is queried in-distribution throughout the run, so its input-side hypothesis in Theorem~\ref{thm:neural} holds along the trajectory. The training set spans both the plant family and the estimates that adaptation is likely to visit: $300$ samples from the direct family with $(\gamma,A)\sim\mathcal U([1.5,5.5]\times[2,6.5])$, and $150$ multi-component samples $\sum_{m=1}^M a_m\,\rho(x,\gamma_x)\rho(\xi_1,\gamma_\xi)\rho(\xi_2,\gamma_\xi)$ with $M\sim\mathcal U\{2,3,4\}$, $a_m\sim\mathcal U(20,200)\cdot\mathcal U\{-1,1\}$, and $\gamma_x,\gamma_\xi\sim\mathcal U(1,6)$. On a $90/10$ split the operator reaches a relative $L^2$ test error of $8\times10^{-4}$; training takes about one hour on an RTX Blackwell GPU, and in closed loop a single forward pass runs $3.8\times$ faster than the numerical kernel solve ($0.17$ vs.\ $0.64$\,ms per step).

\paragraph*{Two time scales, and why the estimates do not reach the truth} The figure separates two time scales. The instability seen in the open-loop panel is, in closed loop, the transient that drives the estimate: $\hat f_2$ moves while the state is large, and $\hat k_2$---the operator's image of $\hat f_2$---moves with it. Both settle early, well before the states, which continue their slower descent to zero on the transport time scale. The separation is structural, and familiar from adaptive control: the estimate is driven only while the state excites it, and here the controller succeeds in extinguishing the state. Once $u\to0$ the excitation is gone and the estimate freezes at whatever value it has reached---not the true $f_2$, and $\hat k_2$ not the true kernel. They need not be: they are good enough that the certainty-equivalence controller built on them regulates the state to zero, which in turn removes the very signal that would have refined them.

\paragraph*{What the operator is, and is not, doing} This separates two stories that are easy to conflate. The non-convergence of $\hat f_2$ to the truth is a property of adaptive control without excitation---it would occur with exact kernels too---and has nothing to do with operator learning. The neural operator's only task, and the only thing the figure asks of it, is to be accurate enough that the kernel it emits from the running estimate matches, in real time, what an exact online cascade solve would have produced. The operator-error panel confirms this: the deviation from the exact kernel falls to about $2\times10^{-2}$ as the estimate stabilizes and holds there. And so the adaptive loop behaves as it would under exact adaptive control with the kernels solved numerically at every step---at a fraction of the online cost.

\section{Conclusions}\label{sec:conclusions}
%==============================================================

We have developed an adaptive controller for a hyperbolic PDE whose nonlinearity is an entire Volterra series with unknown kernels, and we have made it run in real time. These two demands pull against each other about as hard as any pairing in PDE control. Adaptation requires that the feedback gains be recomputed as the plant estimate moves; for this class, however, the gains are not a formula but the solution of a cascade of kernel PDEs whose simplex domains climb in dimension without bound. A controller that must be resynthesized at every instant, out of a computation one would hesitate to run even once---that is the bind, and it is what has kept adaptive nonlinear control in infinite dimension out of reach. The neural operator undoes this bind: trained a single time, offline, it returns the entire gain tuple from one forward evaluation, so a certainty-equivalence law that could previously only be written down now closes the loop.

What is new in the adaptive problem is a difficulty absent from every fixed-model predecessor: the linear kernel is itself unknown. That one fact turns the gain cascade from a chain of explicit quadratures into a self-coupled Volterra integral equation, and we prove it invertible for a nonlinearity of any size, with no smallness condition. Our analysis deals with a target system carrying five vanishing perturbations, all dominated by a single Lyapunov function. Four are interior---the plant--observer mismatch, the observer damping, the cross-order tail the truncated feedback leaves unmatched, and the drift of gains recomputed as the estimate moves---and the fifth acts at the boundary: the surrogate's error. That every one of them vanishes at the origin is what preserves regulation; the boundary term alone scales with the operator's accuracy, and the guaranteed basin recovers the exact-kernel basin as that accuracy improves.

The simulation shows the loop doing what it could not do before: stabilizing the plant while the kernels are still being learned and while the gains are produced on the fly by the offline-trained operator, rather than by a cascade solve that no online implementation could afford. The exact-kernel controller was never the bottleneck---computing it inside the loop was, and that is precisely what the neural operator supplies.

%==============================================================
\section*{Acknowledgment}
%==============================================================
The author thanks Luke Bhan, who declined coauthorship of this primarily theoretical paper. Luke developed the numerical results in Section~\ref{sec:sims}.

%==============================================================
%==============================================================

\begin{IEEEbiography}[{\includegraphics[width=1in,height=1.25in,clip,keepaspectratio]{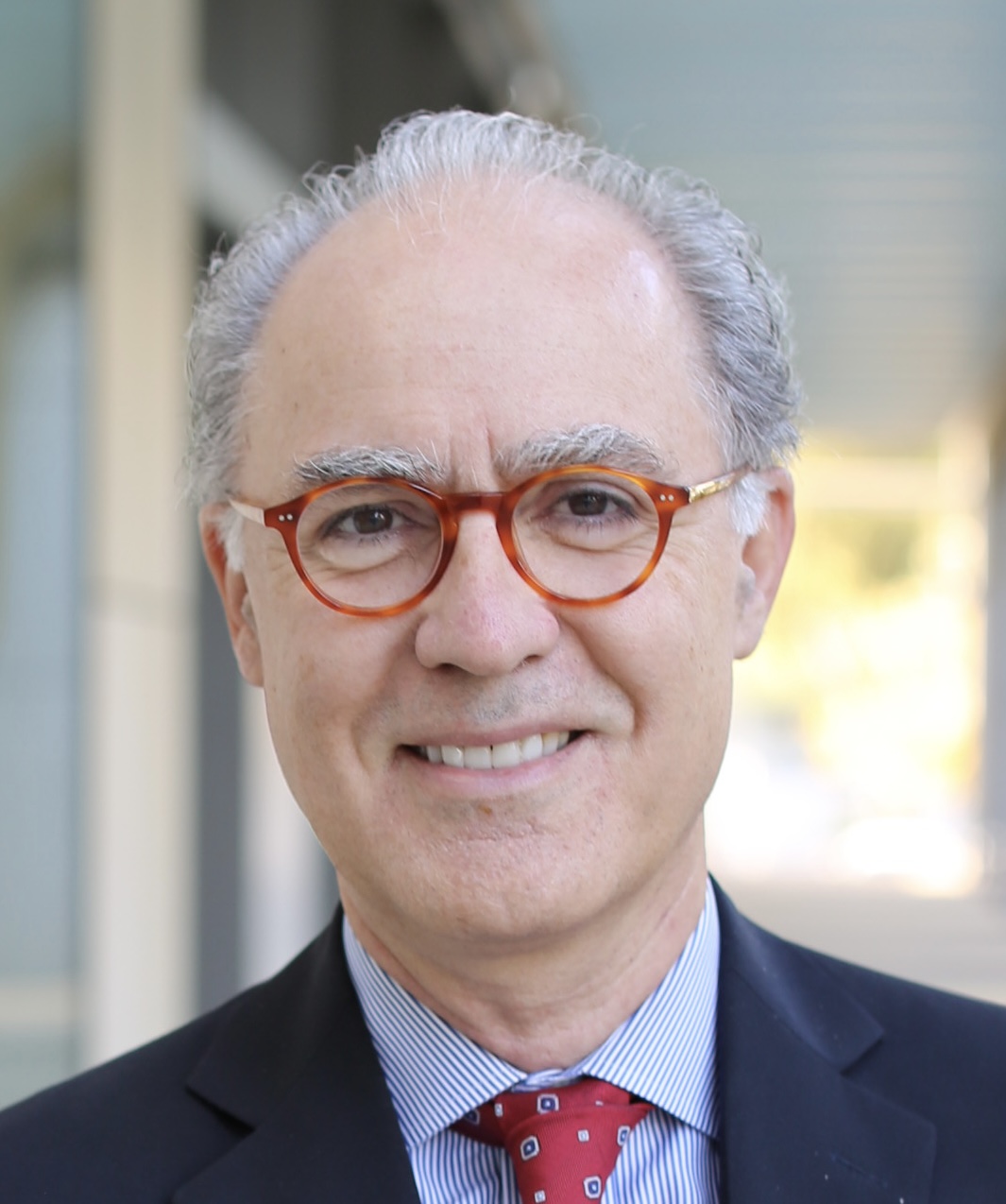}}]{Miroslav Krsti\'c} (Fellow, IEEE) received the Ph.D. degree in electrical engineering from UC Santa Barbara in 1995. 

He is currently a Distinguished Professor with UC San Diego. Dr. Krsti\'c is Fellow of IFAC, SIAM, AIAA, ASME, AAAS, IET, Serbian Academy of Sciences and Arts, and Academia Europaea.

Krsti\'c was the recipient of the IEEE Roger W. Brocket Control Systems Award, Hendrik Bode Lecture Prize, Richard E. Bellman Control Heritage Award, SIAM Reid Prize, ASME Rufus Oldenburger Medal, Harry Nyquist Lecture Prize, Henry Paynter Outstanding Investigator Award, John Ragazzini Education Award, four IFAC TC Awards (Nonlinear Control Systems, Distributed Parameter Systems, Adaptive and Learning Systems, Time-Delay Systems), IFAC Harold Chestnut Control Textbook Award, AV Balakrishnan Award for the Mathematics of Systems, and CSS Distinguished Member Award. Early in his career he received the UC Santa Barbara best dissertation award, the student best paper awards at CDC and ACC, the PECASE, NSF Career, and ONR Young Investigator awards, the Schuck and Axelby paper prizes, and the first UCSD Research Award given to an engineer. 

Krsti\'c is a Fellow-Ambassador of the French CNRS and has also been awarded the Miller Distinguished Visiting Professorship and Springer Visiting Professorship at UC Berkeley, the Distinguished Visiting Fellowship of the Royal Academy of Engineering, and the Invitation Fellowship of the Japan Society for the Promotion of Science.
%, and four honorary professorships outside of the United States. 

He serves as the Editor-in-Chief of IEEE Transactions on Automatic Control and has previously served as Editor-in-Chief of Systems and Control Letters and Senior Editor in Automatica. 
%He has also served as Senior Associate Vice Chancellor for Research at UC San Diego (2012-2026). Krstic has coauthored nineteen books on adaptive, nonlinear, and stochastic control, extremum seeking, control of PDE systems including turbulent flows, and control of delay systems.
\end{IEEEbiography}

\end{document}